\documentclass[10pt,journal,twocolumn,twoside]{IEEEtran}

 \IEEEpubidadjcol 
\IEEEoverridecommandlockouts

\usepackage{graphics} % for pdf, bitmapped graphics files
\usepackage{epsfig} % for postscript graphics files
\usepackage{amsmath} % assumes amsmath package installed
\usepackage{amssymb}  % assumes amsmath package installed

\usepackage{algpseudocode}
\usepackage{algorithm}
\usepackage{epstopdf}

\usepackage{verbatim}

\usepackage{amsthm}
\usepackage{color}
\usepackage{setspace}
\usepackage{multirow}

\newtheorem{assumption}{\bf Assumption}
\newtheorem{definition}{\bf Definition}
\newtheorem{theorem}{\bf Theorem}
\newtheorem{proposition}{\bf Proposition}
\newtheorem{lemma}{\bf Lemma}

\newtheorem{remark}{\bf Remark}

\usepackage{tikz}               
\usepgflibrary{arrows}

\usepackage{hyperref}

\title{A nonlinear model predictive control framework using reference generic terminal ingredients \\- extended version}
\author{Johannes K\"ohler$^1$, Matthias A. M\"uller$^2$, Frank Allg\"ower$^1$
\thanks{$^1$Johannes K\"ohler and Frank Allg\"ower are with the
Institute for Systems Theory and Automatic Control, University of Stuttgart,
70550 Stuttgart, Germany. (email:$\{$johannes.koehler, frank.allgower$\}$@ist.uni-stuttgart.de).}
\thanks{$^2$Matthias A. M\"uller is with the Institute of Automatic Control, Leibniz University Hannover, 30167 Hannover, Germany.
(email:mueller@irt.uni-hannover.de).}
\thanks{Johannes K\"ohler would like to thank the German Research Foundation (DFG) for financial support of the project within the International Research Training Group “Soft Tissue Robotics” (GRK 2198/1 - 277536708).} 
}

\begin{document}
\IEEEoverridecommandlockouts
\IEEEpubid{\begin{minipage}{\textwidth}\ \\[12pt] \\ \\
         \copyright 2020 IEEE.  Personal use of this material is  permitted.  Permission from IEEE must be obtained for all other uses, in  any current or future media, including reprinting/republishing this material for advertising or promotional purposes, creating new  collective works, for resale or redistribution to servers or lists, or  reuse of any copyrighted component of this work in other works.
     \end{minipage}}
\maketitle
\begin{abstract}
In this paper, we present a quasi infinite horizon nonlinear model predictive control (MPC) scheme for tracking of generic reference trajectories. 
This scheme is applicable to nonlinear systems, which are locally incrementally stabilizable.  
For such systems, we provide a reference generic offline procedure to compute an incrementally stabilizing feedback with a continuously parameterized quadratic quasi infinite horizon terminal cost. 
As a result we get a nonlinear reference tracking MPC scheme with a valid terminal cost for general reachable reference trajectories without increasing the online computational complexity. 
As a corollary, the terminal cost can also be used to design nonlinear MPC schemes that reliably operate under online changing conditions, including unreachable reference signals.  
The practicality of this approach is demonstrated with a benchmark example.

This paper is an extended version of the accepted paper~\cite{JK_QINF}, and contains additional details regarding 
 \textit{robust} trajectory tracking (App.~\ref{sec:app_robust}), continuous-time dynamics (App.~\ref{sec:app_cont}),  output tracking stage costs (App.~\ref{sec:app_output}) and the connection to  incremental system properties (App.~\ref{sec:app_increm}). 
\end{abstract}
%!TEX root = ./Tracking_Journal.tex
%%%%%%%%%%%%%%%%%%%%%%%%%%%%%%%%%%%%%%%%%%%%%%%%%%%%%%%%%%%%%%%%%%%%%%%%%%%%%%%
\section{Introduction}
Model Predictive Control (MPC)~\cite{rawlings2017model} is a well established control method, that computes the control input by repeatedly solving an optimization problem online.
The main advantages of MPC are the ability to cope with general nonlinear dynamics, hard state and input constraints, and the inclusion of performance criteria. 
In MPC (theory), recursive feasibility and closed-loop stability of a desirable setpoint are usually ensured by including suitable terminal ingredients (terminal set and terminal cost) in the optimization problem~\cite{mayne2000constrained}. 

In many applications, the control goal goes beyond the stabilization of a pre-determined setpoint. 
These practical challenges include tracking of changing reference setpoints, stabilization of dynamic trajectories, output regulation and general economic optimal operation. 
There exist many promising ideas to tackle these issues in MPC, for example by simultaneously optimizing an artificial reference~\cite{limon2008mpc,limon2016mpc,limon2018nonlinear,fagiano2013generalized,muller2013economic,muller2014performance,ferramosca2014economic}. 
However, most of these approaches are limited in some form to linear systems and/or setpoint stabilization. 
The computation of suitable terminal ingredients seems to be a bottleneck for the practical extension of these methods to nonlinear systems and dynamic trajectories. 
We bridge this gap, by providing a reference \textit{generic} offline computation for the terminal ingredients.  
Thus, we can provide practical schemes for nonlinear systems subject to changing operating conditions.  

\IEEEpubidadjcol

\subsection*{Related work}
For linear stabilizable systems, a terminal set and terminal cost can be computed based on the linear quadratic regulator (LQR) and the maximal output admissible set~\cite{gilbert1991linear}. 
For the purposes of stabilizing a given setpoint, a suitable design procedure for nonlinear systems with a stabilizable linearization has been provided in~\cite{chen1998quasi,rawlings2017model}. 

In practice, the setpoint to be stabilized can change and thus procedures independent of the setpoint are necessary. 
In~\cite{findeisen2000nonlinear}, the issue of finding a setpoint independent terminal cost has been investigated based on the concept of pseudo linearizations. 
While in principle very appealing, the computation of such a pseudo linearization for general nonlinear systems seems unpractical.
In~\cite{magni2005solution}, a locally stabilizing controller is assumed and the terminal cost and constraints are defined implicitly based on the infinite horizon tail cost.  
The main drawback of this method is the implicit description of the terminal cost, which can significantly increase the online computational demand. 
In~\cite{limon2018nonlinear} the feasible setpoints are partitioned into disjoint sets and for each such set a fixed stabilizing controller and terminal cost are designed using the methods in~\cite{wan2003efficient,wan2003offline} based on a local linear time-varying (LTV) system description.  
This method is mainly limited to systems with a one dimensional steady-state manifold, due to the otherwise complex and difficult partitioning. 
In addition, the piece-wise definition can also lead to numerical difficulties since the terminal cost is not differentiable with respect to the setpoint.  
 
There are many applications in which we want to stabilize some dynamic trajectory or periodic orbit. 
The nonlinear system along this trajectory can be locally approximated with an LTV system. 
In~\cite{faulwasser2011model}, this is used to compute a (time-varying) terminal cost for asymptotically constant trajectories. 
In~\cite{aydiner2016periodic} periodic trajectories are considered and a (periodic) terminal cost is computed based on linear matrix inequalities (LMIs).
A significant practical restriction for these methods is the fact that the offline computation is accomplished for a \textit{specific} (a priori known) trajectory.

In general, the existing procedures to compute terminal ingredients for MPC are mainly focused on computing a terminal cost for a \textit{specific} reference point or reference trajectory.
Thus, online changes in the setpoint or trajectory cannot be handled directly and necessitate repeated offline computations.

\subsection*{Contribution}
In this work, we provide a reference \textit{generic} offline procedure to compute a parameterized terminal cost. 
This procedure is applicable to both setpoint or trajectory stabilization. 
The feasibility of this approach requires local incremental stabilizability of the nonlinear dynamics. 
The existing design procedures~\cite{chen1998quasi,faulwasser2011model,aydiner2016periodic} use the linearization around the considered setpoint or trajectory to locally establish properties of the nonlinear systems.  
In a similar spirit, we consider the linearization of the nonlinear system dynamics around all possible points in the constraint set and describe the dynamics analogous to quasi-linear parameter-varying (LPV) systems. 
With this description, we formulate the desired properties on the linearized dynamics and provide suitable LMIs to compute the parameter dependent terminal cost and controller. 
In closed-loop operation we have a quadratic terminal cost with an ellipsoidal terminal constraint directly available.  
This provides a generalization of the offline computations in~\cite{chen1998quasi,faulwasser2011model,aydiner2016periodic} to \textit{generic} references. 
We employ the proposed method in an evasive maneuver test for a car and show that the design of suitable reference generic terminal ingredients can significantly improve the control performance compared to MPC schemes with terminal equality constraints or without terminal constraints.

Given these terminal ingredients, we can extend existing tracking MPC schemes, such as~\cite{limon2008mpc,limon2016mpc,limon2018nonlinear,fagiano2013generalized,muller2013economic,muller2014performance,ferramosca2014economic} to nonlinear system dynamics and optimal periodic operation, which is a fundamental step towards practical nonlinear MPC schemes.
In particular, we provide a nonlinear periodic tracking MPC scheme for exogenous output signals as an extension to~\cite{limon2008mpc,limon2016mpc,limon2018nonlinear}.

\subsection*{Outline}
The remainder of this paper is structured as follows: 
Section~\ref{sec:MPC} presents the reference tracking MPC scheme based on the proposed parameterized terminal ingredients. 
Section~\ref{sec:loc_stab} provides a constructive procedure to design parametric terminal ingredients independent of the considered reference. 
Section~\ref{sec:ext} shows how the resulting parameterized terminal ingredients can be used to extend existing  MPC schemes for changing operation conditions to nonlinear system dynamics and periodic operation. 
Section~\ref{sec:num} shows the practicality of this procedure with numerical examples.
Section~\ref{sec:sum} concludes the paper.  
In the appendix, these results are extended to \textit{robust} trajectory tracking (App.~\ref{sec:app_robust}), continuous-time dynamics (App.~\ref{sec:app_cont}), and output tracking stage costs (App.~\ref{sec:app_output}). 
In addition, the connection between the generic terminal ingredients and incremental system properties is discussed (App.~\ref{sec:app_increm}). 
%!TEX root = ./Tracking_Journal.tex
%%%%%%%%%%%%%%%%%%%%%%%%%%%%%%%%%%%%%%%%%%%%%%%%%%%%%%%%%%%%%%%%%%%%%%%%%%%%%%%
\section{Reference tracking model predictive control}
\label{sec:MPC}
\subsection{Notation}
The quadratic norm with respect to a positive definite matrix $Q=Q^\top$ is denoted by $\|x\|_Q^2=x^\top Q x$.
The minimal and maximal eigenvalue of a symmetric matrix $Q=Q^\top$ is denoted by $\lambda_{\min}(Q)$ and $\lambda_{\max}(Q)$, respectively. 
The identity matrix is $I_n\in\mathbb{R}^{n\times n}$. 
The interior of a set $\mathcal{X}$ is denoted by $\text{int}(\mathcal{X})$.
The vertices of a polytopic set $\Theta$ are denoted by $\theta_i\in\text{Vert}(\Theta)$.  
\subsection{Setup}
We consider the following nonlinear discrete-time system
\begin{align}
\label{eq:sys}
x(t+1)&=f(x(t),u(t))
\end{align}
 with the state $x\in\mathbb{R}^n$, control input $u\in\mathbb{R}^m$, and time step $t\in\mathbb{N}$. 
The extension of the following derivation to continuous-time dynamics is detailed in Appendix~\ref{sec:app_cont}. 
We impose point-wise in time constraints on the state and input
\begin{align}
\label{eq:constraint}
(x(t),u(t))\in \mathcal{Z}, 
\end{align}
with some compact\footnote{%
The derivations can be extended to time-varying constraint sets $\mathcal{Z}(t)$ and dynamics $f(x,u,t)$.
The consideration of non-compact constraint sets may require additional uniformity conditions on the nonlinear dynamics. 
}   set $\mathcal{Z}$. 
We consider the following assumption regarding the reference signal $r=(x_r,u_r)\in\mathbb{R}^{n+m}.$
\begin{assumption}
\label{ass:ref}
The reference signal $r$ satisfies $r(t)\in\mathcal{Z}_r$, $\forall t\geq 0$, with some set $\mathcal{Z}_r\subseteq\text{int}(\mathcal{Z})$. 
Furthermore, the evolution of the reference signal is restricted by $r(t+1)\in\mathcal{R}(r(t))$, with $\mathcal{R}(r)=\{(x_r^+,u_r^+)\in\mathcal{Z}_r|~x_r^+=f(x_r,u_r)\}$. 
\begin{comment} 
There exists a reference constraint set $\mathcal{Z}_r\subset\text{int}(\mathcal{Z})$ and a set valued map $\mathcal{R}:\mathcal{Z}_r\rightarrow\mathcal{Z}_r$, such that the reference signal $r$ satisfies 
\begin{align*}
r(t)\in&\mathcal{Z}_r,\quad \forall t\geq 0,\\
r(t+1)\in&\mathcal{R}(r(t))\subseteq\mathcal{Z}_r.
\end{align*}
\end{comment} 
\end{assumption}
This assumption characterizes that the reference trajectory $r$ is reachable, i.e., follows the dynamics $f$ and lies (strictly) in the constraint set $\mathcal{Z}$. 
If the reference trajectory is not reachable it is possible to enforce these constraints on an artificial reference trajectory which can be included in the MPC optimization problem, compare Section~\ref{sec:ext}.  
\begin{remark}
\label{rk:ref}
The set $\mathcal{R}(r)$ can be modified to incorporate additional incremental input constraints $\|u_r(t+1)-u_r(t)\|_{\infty}\leq \epsilon$. 
Setpoints are included as a special case, with $\mathcal{R}(r)=r$ and the steady-state manifold $\mathcal{Z}_r$. 
\end{remark}

\subsection{Terminal cost and terminal set}
Denote the tracking error by $e_r(t)=x(t)-x_r(t)$. 
The control goal is to stabilize the tracking error $e_r(t)=0$ and achieve constraint satisfaction $(x(t),u(t))\in\mathcal{Z}$, $\forall t\geq 0$.   
To this end we define the quadratic reference tracking stage cost 
\begin{align}
\label{eq:stage}
\ell(x,u,r)=\|x-x_{r}\|_Q^2+\|u-u_{r}\|_R^2,
\end{align} 
with positive definite weighting matrices $Q,~R$.
\begin{remark}
\label{remark:output_cost}
The extension to an tracking stage cost $\ell(x,u,r)=\|h(x,u)-h(x_r,u_r)\|_{S(r)}^2$ with some output $y=h(x,u)$ and a positive definite weighting matrix $S$ is discussed in the Appendix~\ref{sec:app_output}. 
\end{remark}
As discussed in the introduction, we need suitable terminal ingredients to ensure stability and recursive feasibility for the closed-loop system. 
\begin{assumption}
\label{ass:term} 
There exist matrices $K_f(r)\in\mathbb{R}^{m\times n}$, $P_f(r)\in\mathbb{R}^{n\times n}$ with $c_l I_n\leq P_f(r)\leq c_u I_n$, a terminal set $\mathcal{X}_f(r)=\{x\in\mathbb{R}^n|~V_f(x,r)\leq \alpha\}$ with the terminal cost $V_f(x,r)=\|x-x_r\|_{P_f(r)}^2$, such that the following properties hold for any $r\in\mathcal{Z}_r$, any $x\in\mathcal{X}_f(r)$ and any $r^+\in\mathcal{R}(r)$
\begin{subequations}
\label{eq:term}
\begin{align}
\label{eq:term_dec}
V_f(x^+,r^+)\leq& V_f(x,r)-\ell(x,k_f(x,r),r),\\
\label{eq:term_con}
(x,k_f(x,r))\in&\mathcal{Z},
\end{align}
\end{subequations}
with $x^+=f(x,k_f(x,r))$, $k_f(x,r)=u_r+K_f(r)\cdot (x-x_r)$ and positive constants $c_l,~c_u,~\alpha$. 
\end{assumption}
For $r=r^+=0$ this reduces to the standard conditions in~\cite{chen1998quasi}. % in
For a given trajectory $r$, this implies time-varying terminal ingredients, compare~\cite{faulwasser2011model,aydiner2016periodic}. 
Designing suitable\footnote{%
In principle, this assumption can always be satisfied with a terminal equality constraint $\mathcal{X}_f(r)=x_r$.   
However, this can lead to numerical problems, and decrease performance and robustness of the MPC scheme. 
In addition, tracking schemes such as~\cite{limon2008mpc,limon2018nonlinear,kohler2018mpc}, typically require a non-vanishing terminal set size $\alpha$ to ensure exponential stability, compare Section~\ref{sec:ext}. 
} terminal ingredients that satisfy this assumption is the main contribution of this paper and is discussed in more detail in the Section~\ref{sec:loc_stab}. 

\begin{remark}
\label{rk:invarset}
Assumption~\ref{ass:ref} implies that the reference $r(t)$ is contained within a control invariant subset $\mathcal{Z}_{\infty}\subseteq\mathcal{Z}_r$. 
Thus, Assumption~\ref{ass:term} could be relaxed, such that the conditions~\eqref{eq:term} only need to be satisfied for points $r\in\mathcal{Z}_{\infty}$. 
The exact characterization of the set $\mathcal{Z}_{\infty}$ is, however, challenging and thus we consider the stricter\footnote{If there exists a fixed constant $T_0$, such that $r(t+k)\in\mathcal{Z}_r,~\forall k\in[0,T_0]$, implies $r(t)\in\mathcal{Z}_{\infty}$, then the conditions in Assumption~\ref{ass:term} are not stricter. 
However, if we use a convex overapproximation (Prop.~\ref{prop:LMI_lpv}) and/or parameterize the matrices $P_f,~K_f$, then this may introduce additional conservatism.  
} conditions as formulated in Assumption~\ref{ass:term}. 
\end{remark}
\subsection{Preliminary results}
Denote the reference $r$ over the prediction horizon $N$ by  ${r}(\cdot|t)\in\mathbb{R}^{(n+m)\times (N+1)}$ with $r(k|t)=r(t+k)$, $k=0,\dots,N$. 
Given a predicted state and input sequence $x(\cdot|t)\in\mathbb{R}^{n\times N+1 },~u(\cdot|t)\in\mathbb{R}^{m\times N}$ the tracking cost with respect to the reference $r(\cdot|t)$ is given by
\begin{align*}
J_N(x(\cdot|t),u(\cdot|t),r(\cdot|t)):=&\sum_{k=0}^{N-1}\ell(x(k|t),u(k|t),r(k|t))\\
&+V_f(x(N|),r(N|t)).
\end{align*}
The MPC scheme is based on the following (standard) MPC optimization problem
\begin{subequations}
\label{eq:MPC}
\begin{align}
\label{eq:MPC_cost}
V(x(t),r(\cdot|t))=\min_{u(\cdot|t)}&J_N(x(\cdot|t),u(\cdot|t),r(\cdot|t))\\
\label{eq:MPC_dyn}
\text{s.t. }&x(k+1|t)=f(x(k|t),u(k|t)),\\
\label{eq:MPC_init}
&x(0|t)=x(t),\\
\label{eq:MPC_con}
&(x(k|t),u(k|t))\in\mathcal{Z},\\
\label{eq:MPC_term}
&x(N|t)\in\mathcal{X}_f({r}(N|t)).
\end{align}
\end{subequations}
The solution to this optimization problem are the value function $V$ and the optimal input trajectory $u^*(\cdot|t)$. 
In closed-loop operation we apply the first part of the optimized input trajectory to the system, leading to the following closed loop  
\begin{align}
\label{eq:close}
x(t+1)=f(x(t),u^*(0|t))=x^*(1|t),\quad t\geq 0.
\end{align}
The following theorem summarizes the standard theoretical properties of the closed-loop system~\eqref{eq:close}. 
\begin{theorem}
\label{thm:MPC} 
Let Assumptions~\ref{ass:ref} and \ref{ass:term} hold. 
Assume that Problem~\eqref{eq:MPC} is feasible at $t=0$. 
Then Problem~\eqref{eq:MPC} is recursively feasible and the tracking error $e_r=0$ is (uniformly) exponentially stable for the resulting closed-loop system~\eqref{eq:close}. 
\end{theorem}
\begin{proof}
This theorem is a straight forward extension of standard MPC results in~\cite{rawlings2009model}, compare also~\cite{faulwasser2011model}. 
Given the optimal solution $u^*(\cdot|t)$, the candidate sequence 
\begin{align}
\label{eq:candidate_input}
u(k|t+1)=   \begin{cases}
u^*(k+1|t)& k\leq N-2  \\
k_f(x^*(N|t),r(N|t))&k=N-1
\end{cases},
\end{align}
 is a feasible solution to~\eqref{eq:MPC_cost} and implies 
\begin{align}
\label{eq:V}
V(x(t+1),r(\cdot|t+1))\leq V(x(t),r(\cdot|t))-\ell(x(t),u(t),r(t)).
\end{align}
Compact constraints in combination with the quadratic terminal cost imply 
\begin{align*}
\|x(t)-x_r(t)\|_Q^2\leq V(x(t),r(\cdot|t))\leq c_v \|x(t)-x_r(t)\|_Q^2,
\end{align*}
for some $c_v\geq 1$. 
Uniform exponential stability follows from standard Lyapunov arguments using the value function $V$.
\end{proof}
This theorem shows that if we can design suitable terminal ingredients (Ass.~\ref{ass:term}), the closed-loop tracking MPC has all the (standard) desirable properties. 
In Section~\ref{sec:ext} we discuss how this can be extended to more general tracking problems. 
This scheme can be easily modified to ensure robust reference tracking using the method in~\cite{kohler2018novel}, for details see Appendix~\ref{sec:app_robust} and the numerical example in Section~\ref{sec:num}. 
\begin{remark}
\label{rk:withoutterm}
A powerful alternative to the proposed quasi-infinite horizon reference tracking MPC scheme would be a reference tracking MPC scheme without terminal ingredients~\cite{kohlernonlinear19} ($V_f(x,r)=0$, $\mathcal{X}_f(r)=\mathcal{X}$). 
If it is possible to design terminal ingredients (Ass.~\ref{ass:term}), the value function of such an MPC scheme without terminal constraints is locally bounded by $V(x(t),r(\cdot|t))\leq \gamma \ell(x,u,r)$, with a suitable constant $\gamma$, compare~\cite[Prop.~2]{kohlernonlinear19}. 
Thus, an MPC scheme without terminal constraints enjoys similar closed-loop properties to Theorem~\ref{thm:MPC}, provided a sufficiently large prediction horizon $N$ is used, compare~\cite[Thm.~2]{kohlernonlinear19}. 
One of the core advantages of including suitably designed terminal ingredients is that we can implement the MPC scheme with a short prediction horizon $N$. 
On the other hand, if the reference is not reachable (Ass.~\ref{ass:ref}), MPC schemes without terminal constraints can still be successfully applied~\cite[Thm.~4]{kohlernonlinear19}, which is in general not the case for MPC schemes with terminal constraints.  
\end{remark}

%!TEX root = ./Tracking_Journal.tex
%%%%%%%%%%%%%%%%%%%%%%%%%%%%%%%%%%%%%%%%%%%%%%%%%%%%%%%%%%%%%%%%%%%%%%%%%%%%%%%
\section{Reference generic offline computations}
\label{sec:loc_stab}
This section provides a reference  \textit{generic} offline computation to design terminal ingredients for nonlinear reference tracking MPC. 
In Lemma~\ref{lemma:lpv} we provide sufficient conditions for the terminal ingredients based on properties of the linearization. 
Then, two approaches based on LMI computations are described to compute the terminal ingredients, based on Lemma~\ref{lemma:LMI} and Proposition~\ref{prop:LMI_lpv}. 
After that, a procedure to obtain a non conservative terminal set size $\alpha$ is discussed.  
Finally, the overall offline procedure is summarized in Algorithm~\ref{alg:offline}. 
For the special case of setpoint tracking, existing methods are discussed in relation to the proposed procedure. 
In Appendix~\ref{sec:app_cont} and \ref{sec:app_output}, these results are extended to continuous-time dynamics and output tracking stage costs, respectively. 
%
%!TEX root = ./Tracking_Journal.tex
%%%%%%%%%%%%%%%%%%%%%%%%%%%%%%%%%%%%%%%%%%%%%%%%%%%%%%%%%%%%%%%%%%%%%%%%%%%%%%%
\subsection{Sufficient conditions based on the linearization}
We denote the Jacobian of $f$ evaluated around an arbitrary point $r\in\mathcal{Z}_r$ by
\begin{align}
\label{eq:A_r}
A(r)=\left.\left[\dfrac{\partial f}{\partial x}\right]\right|_{(x,u)=r},\quad B(r)=\left.\left[\dfrac{\partial f}{\partial u}\right]\right|_{(x,u)=r}. 
\end{align}
The following lemma establishes local incremental properties of the nonlinear system dynamics based on the linearization. 
\begin{lemma}
\label{lemma:lpv}
Suppose that $f$ is twice continuously differentiable. 
Assume that there exists a matrix $K_f(r)\in\mathbb{R}^{m\times n}$ and a positive definite matrix $P_f(r)\in\mathbb{R}^{n\times n}$ continuous in $r$, such that for any  $r\in\mathcal{Z}_r$, $r^+\in\mathcal{R}(r)$, the following matrix inequality is satisfied
\begin{align}
\label{eq:lpv}
&(A(r)+B(r)K_f(r))^\top P_f(r^+)(A(r)+B(r)K_f(r))\\
\leq& P_f(r)-(Q+K_f(r)^\top R K_f(r))-\epsilon I_n\nonumber
\end{align}
 with some positive constant $\epsilon$. 
Then there exists a sufficiently small constant $\alpha$, such that $P_f,~K_f$ satisfy Assumption~\ref{ass:term}.  
\end{lemma}
\begin{proof}
The proof is very much in line with the result for setpoints in~\cite{chen1998quasi,rawlings2017model}. 
First we show  satisfaction of the decrease condition~\eqref{eq:term_dec} and then constraint satisfaction~\eqref{eq:term_con}. \\
\textbf{Part I:} 
Denote $\Delta x:=x-x_r$ and $\Delta u:=K_f(r)\Delta x$.  
Using a first order Taylor approximation at $r=(x_r,u_r)$, we get
\begin{align*}
f(x,k_f(x,r))={f(x_r,u_r)}+A(r)\Delta x+B(r)\Delta u+\Phi_r(\Delta x),
\end{align*}
with the remainder term $\Phi_r$. 
The terminal cost satisfies
\begin{align}
\label{eq:lpv_1}
 &V_f(x^+,r^+)=\|f(x,u)-f(x_r,u_r)\|_{P_f(r^+)}^2\nonumber\\
=&\|(A(r)+B(r)K_f(r))\Delta x+\Phi_r(\Delta x)\|_{P_f(r^+)}^2\nonumber\\
\leq&\|(A(r)+B(r)K_f(r))\Delta x\|_{P_f(r^+)}^2\nonumber
+\|\Phi_r(\Delta x)\|_{P_f(r^+)}^2\nonumber\\
&+2\|\Phi_r(\Delta x)\|_{P_f(r^+)}\|(A(r)+B(r)K_f(r))\Delta x\|_{P_f(r^+)}\nonumber\\
\stackrel{\eqref{eq:lpv}}{\leq}&V_f(x,r)-\epsilon\|\Delta x\|^2-\ell(x,k_f(x,r),r)+\|\Phi_r(\Delta x)\|_{P_f(r^+)}^2\nonumber\\
&+2\|\Phi_r(\Delta x)\|_{P_f(r^+)}{\|(A(r)+B(r)K_f(r))\Delta x\|_{P_f(r^+)}}.
\end{align}
Using the continuity of $P_f(r),~K_f(r)$ and the compactness of the constraint set $\mathcal{Z}_r$, there exist finite constants
\begin{align}
\label{eq:c_u}
&c_{u}=\max_{r\in\mathcal{Z}_r}\lambda_{\max}(P_f(r)),\quad  c_{l}=\min_{r\in\mathcal{Z}_r}\lambda_{\min}(P_f(r)),\\
\label{eq:k_u}
&k_{u}=\max_{r\in\mathcal{Z}_r}\|K_f(r)\|,\\
&c_{u,2}=\max_{r\in\mathcal{Z}_r}\lambda_{\max}(P_f(r)-(\epsilon I+Q+K_f(r)^\top RK_f(r))  ). \nonumber
\end{align}
Suppose that the remainder term $\Phi_r$ is locally Lipschitz\footnote{
In line with existing procedures~\cite{chen1998quasi}, we first deriving a sufficient local Lipschitz bound $L_{\Phi}^*$ and then obtain a local region $\alpha_1$~\eqref{eq:alpha_1}. 
Alternatively, it is possible to directly use the quadratic bound $\|\Phi_r(\Delta x)\|\leq c\|\Delta x\|^2$ and work with higher order terms to obtain $\alpha_1$, compare~\cite[Prop.~1]{kohlernonlinear19}. 
} continuous in the terminal set with a constant $L_{\Phi,\alpha}$ satisfying
\begin{align}
\|\Phi_r(\Delta x)\|\leq L_{\Phi,\alpha}\|\Delta x\|,\nonumber\\
\label{eq:lpv_2}
L_{\Phi,\alpha}\leq L_{\Phi}^*:=\sqrt{\dfrac{c_{u,2}+\epsilon}{c_u}}-\sqrt{\dfrac{c_{u,2}}{c_u}}.
\end{align}
Then we have
\begin{align*}
&~~\|\Phi_r(\Delta x)\|_{P_f(r^+)}^2\\
&~~+2\|\Phi_r(\Delta x)\|_{P_f(r^+)}\|(A(r)+B(r)K_f(r))\Delta x\|_{P_f(r^+)}\\
&\stackrel{\eqref{eq:lpv}\eqref{eq:c_u}\eqref{eq:lpv_2}}{\leq} \left(L_{\Phi,\alpha}^2c_u+2L_{\Phi,\alpha} \sqrt{c_u}\sqrt{c_{u,2}}\right)\|\Delta x\|^2\\
&~=~\left(c_u\left(L_{\Phi,\alpha}+\sqrt{\frac{c_{u,2}}{c_u}}\right)^2-c_{u,2}\right)\|\Delta x\|^2\stackrel{\eqref{eq:lpv_2}}{\leq}\epsilon \|\Delta x\|^2,
\end{align*}
which in combination with~\eqref{eq:lpv_1} implies the desired inequality~\eqref{eq:term_dec}. 
Twice continuous differentiability of $f$ in combination with compactness of $\mathcal{Z}$ implies that there exists some constant $T$ with
\begin{align*}
\|\Phi_r(\Delta x)\|\leq T\left(\|\Delta x\|^2+\|\Delta u\|^2\right)\stackrel{\eqref{eq:k_u}}{\leq} T(1+k_u^2)\|\Delta x\|^2. 
\end{align*}
Using  $\|\Delta x\|\leq \sqrt{\frac{\alpha}{c_l}}$ from the terminal constraint, we get~\eqref{eq:lpv_2} for all $\alpha\leq \alpha_1$ with 
\begin{align}
\label{eq:alpha_1}
\alpha_1:=c_l\left(\dfrac{L_{\Phi}^*}{T(1+k_u^2)}\right)^2.
\end{align}
\textbf{Part II:} Constraint satisfaction: 
The terminal constraint $\|\Delta x\|_{P_f(r)}^2\leq \alpha$ in combination with~\eqref{eq:c_u},~\eqref{eq:k_u} implies  
\begin{align*}
(\Delta x,~\Delta u)\in\mathcal{B}(\alpha)=\left\{z\in\mathbb{R}^{n+m}|~\|z\|^2\leq \frac{\alpha}{c_l}\left(1+k_u^2\right)\right\}.
\end{align*}
Given $\mathcal{Z}_r\subseteq\text{Int}(\mathcal{Z})$, there exists a small enough $\alpha_2$ such that 
\begin{align}
\label{eq:alpha_2}
(x,u)=r+(\Delta x,\Delta u)\subseteq \mathcal{Z}_r\oplus\mathcal{B}(\alpha)\subseteq\mathcal{Z},~ \forall \alpha\leq \alpha_2. 
\end{align}
\end{proof}
As a summary, given matrices $P_f,~K_f$ satisfying~\eqref{eq:lpv}, we can compute a local Lipschitz bound~\eqref{eq:lpv_2}, which in turn implies a maximal terminal set size $\alpha_1$. 
Similarly, the constraint sets $\mathcal{Z}$ and $\mathcal{Z}_r$ in combination with $K_f,~P_f$ imply an upper bound $\alpha_2$ to ensure constraint satisfaction. 
Then Assumption~\ref{ass:term} is satisfied for any $\alpha\leq \min\{\alpha_1,~\alpha_2\}$. 
This result is an extension of~\cite{chen1998quasi,rawlings2017model} to arbitrary dynamic references. 
%!TEX root = ./Tracking_Journal.tex
%%%%%%%%%%%%%%%%%%%%%%%%%%%%%%%%%%%%%%%%%%%%%%%%%%%%%%%%%%%%%%%%%%%%%%%%%%%%%%%
\subsection{Quasi-LPV based procedure}
Lemma~\ref{lemma:lpv} states that matrices satisfying inequality~\eqref{eq:lpv} also satisfy Assumption~\ref{ass:term} with a suitable terminal set size $\alpha$. 
In the following, we formulate computationally tractable optimization problems to compute matrices that satisfy the conditions in Lemma~\ref{lemma:lpv}.  
The following Lemma transforms the conditions in~\eqref{eq:lpv} to be linear in the arguments. 
\begin{lemma}
\label{lemma:LMI}
Suppose that there exists matrices $X(r)$,~$Y(r)$ continuous in $r$, that satisfy the constraints in~\eqref{eq:LMI} for all $r\in\mathcal{Z}_r,~r^+\in\mathcal{R}(r)$. 
Then  $P_f=X^{-1}$, $K_f=YP_f$ satisfy~\eqref{eq:lpv}.
\end{lemma}
\begin{proof}
The proof is standard, compare~\cite{boyd1994linear} and Lemma~~\ref{lemma:LMI_output} in the Appendix. 
\begin{comment}
Define $X(r)=P_f(r)^{-1}$ and multiply~\eqref{eq:lpv} from left and right with $X(r)$
\begin{align*}
(A(r)X(r)+B(r)Y(r))^\top X(r^+)^{-1} (A(r)X(r)+B(r)Y(r))-X(r)+X(r)QX(r)+Y(r)^\top RY(r)\leq 0.
\end{align*}
This can be equivalently written as
\begin{align*}
X(r)-
\begin{pmatrix}
A(r)X(r)+B(r)Y(r)\\Q^{1/2}X(r)\\R^{1/2}Y(r)
\end{pmatrix}^{\top}
\begin{pmatrix}
X(r^+)&0&0\\
0&I&0\\
0&0&I
\end{pmatrix}
\begin{pmatrix}
A(r)X(r)+B(r)Y(r)\\Q^{1/2}X(r)\\R^{1/2}Y(r)
\end{pmatrix}
\geq 0
\end{align*}
Using the Schur complement this reduces to~\eqref{eq:LMI}, which is linear in the matrices $X,~Y$. 
\end{comment}
\end{proof}
The optimization problem~\eqref{eq:LMI} is convex, linear in $X,~Y$ and minimizes the worst-case terminal cost $P_f(r)\leq X_{\min}^{-1}$.  
So far, the result is only conceptual, since~\eqref{eq:LMI} is an infinite programming problem (infinite dimensional optimization variables with infinite dimensional constraints). 
In particular, we need a finite parameterization of $X,~Y$ and the infinite constraints need to be converted into a finite set of sufficient constraints. 

\begin{remark}
\label{rk:SOS}
One solution to this problem would be sum-of-squares (SOS) optimization~\cite{parrilo2003semidefinite}. 
Assuming $A,~B$ are polynomial, consider matrices $X,~Y$ polynomial in $r$ (with a specified order $d$) and ensure that the matrix in~\eqref{eq:LMI} is SOS. 
A similar approach is suggested in~\cite{manchester2017control} to find a control contraction metric (CCM) for continuous-time systems (which is a strongly related problem). 
This approach is not pursued here since most systems require a polynomial of high order to approximate the nonlinear dynamics and the computational complexity grows exponentially in $n^d$, thus prohibiting the practical application. 
The connection between CCM and LPV gain-scheduling design is discussed in~\cite{wang2019comparison}.
\end{remark}

We approach this problem from the perspective of quasi-LPV systems and gain-scheduling~\cite{rugh2000research}. 
First, write the Jacobian~\eqref{eq:A_r} as 
\begin{align}
\label{eq:A_paramlin}
A(r)=A_0+\sum_{j=1}^{p} \theta_j(r) A_j,~ 
B(r)=B_0+\sum_{j=1}^{p} \theta_j(r) B_j, 
\end{align}
with some nonlinear (continuously differentiable) parameters  $\theta\in\mathbb{R}^p$.  
This can always be achieved with $p\leq n(n+m)$. 
We impose the same structure on the optimization variables with
\begin{align}
\label{eq:X}
X(r)=X_0+\sum_{j=1}^p \theta_j(r) X_j,~ Y(r)=Y_0+\sum_{j=1}^p \theta_j(r) Y_j. 
\end{align}
\begin{remark}
\label{rk:affine}
For input affine systems of the form $f(x,u)=f_x(x)+B u$, the Jacobian~\eqref{eq:A_paramlin} and correspondingly the parameters $\theta_i$ only depend on $x_r$. 
Thus, the resulting terminal ingredients are solely parameterized by the state $x_r$.  
\end{remark}
Using the parameterization \eqref{eq:A_paramlin}-\eqref{eq:X}, \eqref{eq:LMI} contains only a finite number of optimization variables, but still needs to be verified for all $r\in\mathcal{Z}_r,~r^+\in\mathcal{R}(r)$. 
There are two options to deal with this:  convexifying the problem or gridding the constraint set.

\subsubsection{Convexify}
\begin{table*}
\small{
\begin{subequations}
\label{eq:LMI}
\begin{align}
\min_{X(r),Y(r),X_{\min}}& - \log \det X_{\min}\\
\text{s.t. }&\begin{pmatrix}
X(r)&X(r)A(r)^\top+Y(r)^\top B(r)^\top&(Q+\epsilon)^{1/2}X(r)&(R^{1/2}Y(r))^\top\\
*&X(r^+)&0&0\\
*&*&I&0\\
*&*&*&I
\end{pmatrix}\geq 0,\\
&X_{\min}\leq X(r),\\
\label{eq:LMI_r}
&\forall r\in\mathcal{Z}_r,~r^+\in\mathcal{R}(r).  
\end{align}
\end{subequations}
\begin{subequations}
\label{eq:LMI_LPV}
\begin{align}
\min_{X_i,Y_i,\Lambda_i,X_{\min}}& - \log\det X_{\min}\\
\label{eq:LMI_LPV1}
\text{s.t. }&\begin{pmatrix}
X(\theta)&X(\theta)A(\theta)^\top+Y(\theta)^\top B(\theta)^\top&(Q+\epsilon)^{1/2}X(\theta)&(R^{1/2}Y(\theta))^\top\\
*&X(\theta^+)&0&0\\
*&*&I&0\\
*&*&*&I
\end{pmatrix}
- \begin{pmatrix}\sum_{i=1}^p\theta_i^2\Lambda_i&0\\0&0\end{pmatrix}\geq 0,\\
&X_{\min}\leq X(\theta),\quad \forall (\theta,\theta^+)\in\text{Vert}(\overline{\Theta}),  \\
\label{eq:LMI_LPV2}
&\begin{pmatrix}
0&(A_iX_i+B_iY_i)^\top\\
(A_iX_i+B_iY_i)&0
\end{pmatrix}-\Lambda_i\leq 0,\quad \Lambda_i\geq 0,\quad  i=1,\dots,p.
\end{align}
\end{subequations}
}
\end{table*}
In order to convexify~\eqref{eq:LMI}, we match the constraint sets $\mathcal{Z}_r,~\mathcal{R}(r)$ on the reference $r$ to polytopic constraint sets $\Theta,~\Omega$ on the parameters $\theta$. 
The polytopic sets $\Theta,~\Omega(\theta)$ need to satisfy
\begin{align}
\label{eq:Theta_set}
\theta(r)\in&\Theta,\quad \forall r\in\mathcal{Z}_r,\\
\theta(r^+)\in&\Omega(\theta(r)),\quad \forall r^+\in\mathcal{R}(r). \nonumber
\end{align}	
Computing a set $\Theta$, such that $\theta(r)\in\Theta$ for all $r\in\mathcal{Z}_r$ can be achieved by considering a hyperbox $\Theta=\{\theta\in\mathbb{R}^p|~\theta_i\in[\underline{\theta}_i,\overline{\theta}_i]\}$. 
For $\Omega$, a simple approach is $\Omega(\theta)=\{\theta\}\oplus\Omega$, where $\Omega$ is a hyperbox that encompasses the maximal change in the parameters $\theta$ in one time step, i.e. $\Omega=\{\Delta \theta\in\mathbb{R}^p|~\Delta\theta_i\in[\underline{v}_i,\overline{v}_i]\}$. 
We denote the joint polytopic constraint set by 
\begin{align}
\label{eq:overline_theta}
(\theta,\theta^+)\in\overline{\Theta}=\{(\theta,\theta^+)\in \Theta\times\Theta|~\theta^+\in\{\theta\}\oplus\Omega\},
\end{align}
which consists of $6^p$ vertices. 
The following proposition provides a simple convex procedure to compute a terminal cost, by solving a finite number of LMIs.
\begin{proposition}
\label{prop:LMI_lpv}
Suppose that there exist matrices $X_i,~Y_i,~\Lambda_i,~X_{\min}$ that satisfy the constraints in~\eqref{eq:LMI_LPV}. 
Then the matrices
\begin{align*}
P_f(r)=&X^{-1}(r),\quad K_f(r)=Y(r)P_f(r),
\end{align*} 
satisfy~\eqref{eq:lpv}, with $X,~Y$ according to~\eqref{eq:X}. 
\end{proposition}
\begin{proof}
Due to Lemma~\ref{lemma:LMI}, it suffices to show that $X(r),~Y(r)$ satisfy the constraints in~\eqref{eq:LMI}.  
Due to the definition of the set $\overline{\Theta}$~\eqref{eq:overline_theta} and $\Lambda_i\geq 0$, any solution that satisfies the constraints~\eqref{eq:LMI_LPV1} over all $(\theta,\theta^+)\in\overline{\Theta}$, also satisfies the constraints~\eqref{eq:LMI} for all $r\in\mathcal{Z}_r,~r^+\in\mathcal{R}(r)$. 
It remains to show that it suffices to check the inequality on the vertices of the constraint set $\overline{\Theta}$. 
This last result is a consequence of multi-convexity~\cite[Corollary 3.2]{apkarian2000parameterized}. 
In particular, if a function $f$ is multi-concave along the edges of the constraint set $\overline{\Theta}$, then it attains its minimum at a vertex of $\overline{\Theta}$ and thus it suffices to verify~\eqref{eq:LMI_LPV1} over the vertices of $\overline{\Theta}$.  
The edges of $\overline{\Theta}$~\eqref{eq:overline_theta} are characterized by $\{\theta_i,~\theta_i^+,~\theta_i^+-\theta_i\}$, $i=1,\dots,p$. 
A function is multi-concave if the second derivative w.r.t. these directions is negative-semi-definite, compare~\cite[Corollary~3.4]{apkarian2000parameterized}. 
Similar to~\cite[Corollary~3.5]{apkarian2000parameterized}, the additional constraint~\eqref{eq:LMI_LPV2} ensures that the function is multi-concave. 
Thus, it suffices to verify inequality~\eqref{eq:LMI_LPV1} on the vertices of the constraint set $\overline{\Theta}$.  
\end{proof}
\begin{remark}
\label{remark:box}
The result in Proposition~\ref{prop:LMI_lpv} remains valid, if the set $\overline{\Theta}$ in~\eqref{eq:overline_theta} is replaced by the set $\overline{\Theta}=\Theta\times (\Theta\oplus\Omega)$. 
This set has only $4^p$ vertices and the induced conservatism of this approximation is negligible if $\Omega$ is small compared to $\Theta$. 
\end{remark}

\subsubsection{Gridding} 
A common heuristic to ensure that parameter dependent LMIs such as~\eqref{eq:LMI} hold for all $(r,r^+)$ is to consider the constraints on sufficiently many sample points in the constraint set, compare e.g. \cite[Sec.~4.2]{apkarian2000parameterized}.  
Due to continuity, the constraint is typically satisfied on the full constraint set if it holds on a sufficiently fine grid.
For this method it is crucial that satisfaction of~\eqref{eq:term_dec} is verified by using a fine grid (compare Algorithm~\ref{alg:offline_alpha}).  

The gridding consists of a grid over all possible state and input combinations $(r,r^+)$, i.e., all considered points satisfy 
\begin{align}
\label{eq:grid_r}
r,~r^+\in\mathcal{Z}_r, \quad r^+\in\mathcal{R}(r),\quad \mathcal{R}(r^+)\neq\emptyset.  
\end{align}
For the simple structure $\mathcal{R}(r)$ in Assumption~\ref{ass:ref} this can be achieved by gridding $r$, computing $x_r^+=f(x_r,u_r)$, and considering all $u_r^+$, such that $(x_r^+,u_r^+)\in\mathcal{Z}_r$ and $(f(x_r^+,u_r^+),\tilde{u}_r)\in\mathcal{Z}_r$ with some $\tilde{u}_r$. 
This approach does not introduce additional conservatism, but is computationally challenging for high dimensional systems.  
As discussed in Remark~\ref{rk:ref} we can include additional constraints on the reference, which makes the offline computation less conservative.  
If some parameters, e.g. $u_r$, enter the LMIs affinely and are subject to polytopic constraints, it suffices to consider the vertices of the corresponding constraint set.

The advantage of the convex procedure (compared to the gridding) is that it typically scales better with the system dimension. 
This comes at the cost of additional conservatism due to the construction of the set $\overline{\Theta}$ and the additional multi-convexity constraint~\eqref{eq:LMI_LPV2}. 
The computational demand can be reduced by considering (block-)diagonal multipliers $\Lambda_i=\lambda_i I$. 
It can often be beneficial to consider a combination of the two approaches, i.e. grid in some dimensions and conservatively convexify in others. 
The advantages and applicability of both approaches are explored in more detail in the numerical examples in Section~\ref{sec:num}. 

The main result is that we can formulate the offline design procedure similar to the gain scheduling synthesis of (quasi)-LPV systems and thus can draw on a well established field to formulate\footnote{%
If the parameters $\theta_i$ are chosen based on a vertex representation ($\theta_i\geq 0,\sum_{i=1}^p\theta_i=1$) the multi-convexity condition~\eqref{eq:LMI_LPV2} can be replaced by positivity conditions of the polynomials,  compare for example~\cite{montagner2005gain}.
In~\cite{mao2003robust} a convexification with an additional matrix is considered. 
 More elaborate methods to formulate LPV synthesis with finite LMIs can be found in~\cite{scherer2000linear}.  
} offline LMI procedures, compare~\cite{rugh2000research}.

%!TEX root = ./Tracking_Journal.tex
%%%%%%%%%%%%%%%%%%%%%%%%%%%%%%%%%%%%%%%%%%%%%%%%%%%%%%%%%%%%%%%%%%%%%%%%%%%%%%%
\subsection{Non-conservative terminal set size $\alpha$}  
\label{sec:alpha}
The terminal set size $\alpha$ derived in Lemma~\ref{lemma:lpv} can be quite conservative. 
In the following we illustrate how a non conservative value $\alpha$ can be computed (given $P_f$ and $K_f$). 

\subsubsection{Constraint satisfaction - $\alpha_2$}
Assume that we have polytopic constraints of the form $\mathcal{Z}=\{r=(x,u)|L_r r\leq l\}$. 
The constant $\alpha_2$, with the property that $\alpha\leq \alpha_2$ implies constraint satisfaction~\eqref{eq:term_con}, can be computed with 
\begin{align}
\label{eq:alpha_2_better}
\alpha_2&:=\max_{\alpha}~ \alpha\\
\text{s.t. }& \|P_f(r)^{-1/2}\begin{pmatrix}I_n&K_f^\top(r)\end{pmatrix}L_{r,j}^\top\|^2\alpha\leq (l_j-L_{r,j} r)^2,\nonumber\\ 
& \forall r\in\mathcal{Z}_r,\quad  j=1,\dots n_z. \nonumber 
\end{align}
This problem can be efficiently solved by girdding the constraint set $\mathcal{Z}_r$, solving the resulting linear program (LP) for each point $r$ and taking the minimum.    
In the special case that $P_f,~K_f$ are constant this reduces to one small scale LP.  
\subsubsection{Local Stability - $\alpha_1$}
Determining a non-conservative constant $\alpha_1$, related to the local Lyapunov function $V_f$ can be significantly more difficult. 
For comparison, in the setpoint stabilization case a non-convex optimization problem is formulated to check whether~\eqref{eq:term_dec} holds for a specific value of $\alpha_1$, compare~\cite[Rk.~3.1]{chen1998quasi}.  
In a similar fashion, we consider the following algorithm\footnote{%
Algorithm~\ref{alg:offline_alpha} can be thought of as a sampling based strategy to solve this non-convex optimization problem considered in~\cite[Rk.~3.1]{chen1998quasi}. 
Using standard convex solvers, like sequential quadratic programming (SQP), yield a faster solution, but can get stuck in local minima. 
This is dangerous for this problem, since the local minima correspond to values $\alpha$ that do not satisfy Assumption~\ref{ass:term}.  
Alternatively, nonlinear Lipschitz-like bounds can be used to reduce the conservatism, compare~\cite{griffith2018robustly} (which, however, also use sampling).  
} to determine whether~\eqref{eq:term_dec} holds for all $\alpha\leq \alpha_1$:
  \begin{algorithm}[H]
\caption{Offline computation - Local stability $\alpha_1$}
\label{alg:offline_alpha}
\begin{algorithmic}[1]
\Statex Given a candidate constant $\alpha_1$:
\Statex \textbf{Grid: } {Select $(r,r^+)$ satisfying~\eqref{eq:grid_r}}
\State  Evaluate $P_f(r),P_f(r^+),K_f(r)$ using~\eqref{eq:X}. 
\State Generate random vectors $\Delta x_i$: with $\|\Delta x_i\|_{P_f(r)}^2\leq \alpha_1$.
\State Check if $x_i=x_r+\Delta x_i$ satisfies~\eqref{eq:term_dec}. 
\end{algorithmic}
\end{algorithm}
Starting with $\alpha_1=\alpha_2$, the value $\alpha_1$ is iteratively decreased until all considered combination ($r,r^+,x_i$) satisfy~\eqref{eq:term_dec}.

The overall offline procedure to compute the terminal ingredients (Ass.~\ref{ass:term}) is summarized as follows:
\begin{algorithm}[H]
\caption{Offline computation}
\label{alg:offline}
\begin{algorithmic}[1]
\State  Define $\theta$ corresponding to the linearization~\eqref{eq:A_paramlin}. 
\State  LMI computation using gridding or convexification:
\Statex \textbf{Convex}: Determine hyperbox sets $\Theta$, $\Omega$ satisfying~\eqref{eq:Theta_set}.
\Statex \quad \quad Solve~\eqref{eq:LMI_LPV} using $\overline{\Theta}$ according to~\eqref{eq:overline_theta} or Remark~\ref{remark:box}.  
\Statex \textbf{Gridding}: Select $(r_i,r_i^+)$ satisfying~\eqref{eq:grid_r}. 
\Statex \quad\quad Solve~\eqref{eq:LMI} for all $(r_i,~r_i^+)$. 
\State  Compute size of the terminal set $\alpha=\min\{\alpha_1,\alpha_2\}$:
\Statex \quad a):compute $\alpha_1$ using Algorithm~\ref{alg:offline_alpha} (or~\eqref{eq:alpha_1}),
\Statex \quad b):compute $\alpha_2$ using~\eqref{eq:alpha_2_better} (or~\eqref{eq:alpha_2}).
\end{algorithmic}
\end{algorithm}
The presented offline procedure is considerably more involved than for example the computation for one specific setpoint~\cite{chen1998quasi}. 
We emphasize that this procedure only has to be completed once and we need \textit{no} repeated offline computations to account for changing operation conditions. 
Furthermore, the applicability to nonlinear systems with the corresponding computational effort offline is detailed with numerical examples in Section~\ref{sec:num}. 
%!TEX root = ./Tracking_Journal.tex
%%%%%%%%%%%%%%%%%%%%%%%%%%%%%%%%%%%%%%%%%%%%%%%%%%%%%%%%%%%%%%%%%%%%%%%%%%%%%%%
\subsection{Setpoint tracking}
\label{sec:increm_setpoint}
Now we discuss setpoint tracking, which is included in the previous derivation as a special case with $\mathcal{Z}_r$ such that $(x_r,u_r)\in\mathcal{Z}_r$ implies $x_r=f(x_r,u_r)$ and $\mathcal{R}(r)=r$. 
Note, that both presented approaches significantly simplify in this case. 
For the gridding approach it suffices to grid along the steady-state manifold $\mathcal{Z}_r$ which is typically low dimensional. 
In the convex approach (Prop.~\ref{prop:LMI_lpv}) we have $\theta^+=\theta$ and thus we only consider the $2^p$ vertices of $\Theta$. 

Compared to the dynamic reference tracking problem, the problem of tracking a setpoint has received a lot of attention in the literature and many solutions have been suggested. 

One of the first attempts to solve this issue is the usage of a pseudo linearization in~\cite{findeisen2000nonlinear}. 
There, a nonlinear state and input transformation is sought, such that the linearization of the transformed system around the setpoints is constant and thus constant terminal ingredients can be used. 
This approach seems unpractical, since there is no easy or simple method to compute such a pseudo linearization.
 
In~\cite{limon2018nonlinear,wan2003efficient,wan2003offline} the steady-state manifold $\mathcal{Z}_r$ is partitioned into sets. 
In each set the nonlinear system is described as an LTV system and a constant terminal cost and controller are computed. 
Correspondingly, in closed-loop operation under changing setpoints~\cite{limon2018nonlinear} the terminal cost matrix $P_f$ is piece-wise constant. 
This might cause numerical problems in the optimization, since the cost is not differentiable with respect to the reference $r$. 
Furthermore, the (manual) partitioning of the steady-state manifold seems difficult for general MIMO systems (if the dimension of the steady-state manifold is larger than one). 
In comparison, Algorithm~\ref{alg:offline} yields continuously parameterized terminal ingredients, thus avoiding the need for user defined partitioning and piece-wise definitions. 

In~\cite[Remark~8]{muller2014performance} it was proposed to compute a continuously parameterized controller $K_f(r)$ by analytically using a pole-placement formula and solving the corresponding Lyapunov\footnote{%
In~\cite{muller2014performance}, the terminal cost $V_f$ is computed for a (differentiable) economic stage cost $\ell(x,u)$ (not necessarily quadratic), compare also~\cite{fagiano2013generalized}. 
The computation of the terminal cost is decomposed into a linear and quadratic term, compare~\cite{amrit2011economic}. 
Computing the quadratic term of this economic terminal cost is equivalent to computing a quadratic terminal cost for a quadratic stage cost (Ass~\ref{ass:term}). 
} equation to obtain $P_f(r)$. 
The resulting terminal ingredients are quite similar to the proposed ones. 
However, this procedure cannot be directly translated into a simple optimization problem and might hence not be tractable.

%!TEX root = ./Tracking_Journal.tex
%%%%%%%%%%%%%%%%%%%%%%%%%%%%%%%%%%%%%%%%%%%%%%%%%%%%%%%%%%%%%%%%%%%%%%%%%%%%%%%
\section{Nonlinear MPC subject to changing operation conditions}
\label{sec:ext}
Many control problems are more general than the reference tracking considered in Section~\ref{sec:MPC}. 
One challenge includes tracking and output regulation with exogenous signals in order to accommodate online changing operation conditions. 
For this set of problems, the reference $r$ might not satisfy Assumption~\ref{ass:ref} (due to sudden changes and unreachable signals), compare~\cite{limon2008mpc,limon2016mpc,limon2018nonlinear}. 
More generally, the minimization of a possibly online changing and non-convex economic cost is a (non-trivial) control problem which is often encountered, compare~\cite{fagiano2013generalized,muller2013economic,muller2014performance,ferramosca2014economic}. 
One promising method to solve these problems is the simultaneous optimization of an artificial reference, as done in~\cite{limon2008mpc,limon2016mpc,limon2018nonlinear,fagiano2013generalized,muller2013economic,muller2014performance,ferramosca2014economic}. 
Compared to a standard reference tracking MPC formulation such as~\eqref{eq:MPC}, these schemes ensure recursive feasibility despite changes in exogenous signals (such as the desired output reference or the economic cost).
In this section, we show how the reference generic terminal ingredients can be used to design nonlinear MPC schemes that reliably operate under changing operating conditions, as an extension and combination of the ideas in~\cite{limon2008mpc,limon2016mpc,limon2018nonlinear,fagiano2013generalized,muller2013economic,muller2014performance,ferramosca2014economic}. 
In particular, we present a scheme that exponentially stabilizes the periodic trajectory which best tracks an exogenous output signal.   
The extension of the economic MPC schemes~\cite{fagiano2013generalized,muller2013economic,muller2014performance,ferramosca2014economic} to periodic artificial trajectories based on the reference generic terminal ingredients is beyond the scope of this work and part of current research. 
%!TEX root = ./Tracking_Journal.tex
%%%%%%%%%%%%%%%%%%%%%%%%%%%%%%%%%%%%%%%%%%%%%%%%%%%%%%%%%%%%%%%%%%%%%%%%%%%%%%%
\subsection{Nonlinear periodic tracking MPC subject to changing exogenous output references}
\label{sec:gen_output}
We assume that at time $t$ an exogenous $T$-periodic output reference signal $y_{e}(\cdot|t)\in\mathbb{R}^{p\times T}$ is given. 
For some $T$-periodic reference $r(\cdot|t)=(x_r(\cdot|t),u_r(\cdot|t))\in\mathbb{R}^{(n+m)\times T}$, we define the tracking cost with respect to this output signal $y_e$ by
\begin{align*}
J_T(r(\cdot|t),y_e(\cdot|t))=&\sum_{j=0}^{T-1} \|\underbrace{h(r(j|t))}_{={y_r(j|t)}}-y_e(j|t)\|^2,
\end{align*}
with a bounded nonlinear output function $h:\mathcal{Z}_r\rightarrow\mathbb{R}^p$. 
The objective is to stabilize the feasible $T$-periodic reference trajectory $r$, that minimizes $J_T$. 
In~\cite{limon2008mpc,limon2018nonlinear} the issue of stabilizing the optimal setpoint for piece-wise constant output signals has been investigated.  
In~\cite{limon2016mpc} periodic trajectories have been considered for the special case of linear systems. 
By combining these methods with the proposed terminal ingredients, we can design a nonlinear MPC scheme that stabilizes the optimal periodic\footnote{
In the case of setpoint tracking ($T=1$), the MPC scheme reduces to~\cite{limon2018nonlinear}. 
 As discussed in Section~\ref{sec:increm_setpoint}, the proposed procedure can be used to design suitable terminal ingredients for setpoints. 
} trajectory for periodic output reference signals, compare~\cite{kohler2018mpc}. 
The scheme is based on the following optimization problem
\begin{subequations}
\label{eq:limon}
\begin{align}
&W_T(x(t),y_e(\cdot|t))\\
=&\min_{u(\cdot|t),r(\cdot|t)}J_N(x(\cdot|t),u(\cdot|t),r(\cdot|t))
+J_T(r(\cdot|t),y_e(\cdot|t))\nonumber\\
\text{s.t. }&x(k+1|t)=f(x(k|t),u(k|t)),\quad x(0|t)=x(t),\\
&(x(k|t),u(k|t))\in\mathcal{Z},\quad 
x(N|t)\in\mathcal{X}_f({r}(N|t)),\\
&r(j+1|t)\in\mathcal{R}(r(j|t))\subseteq\mathcal{Z}_r,\\
&r(l+T|t)=r(l|t),  ~ l=0,\dots,\max\{0,N-T\},\\
&j=0,\dots,T-1,\quad k=0,\dots,N-1. \nonumber
\end{align}
\end{subequations}
This scheme is recursively feasible, independent of the output reference signal $y_e$. 
Furthermore, if the exogenous signal $y_e$ is $T$-periodic the closed-loop system is stable. 
Additionally, if a convexity and continuity condition on the set of feasible periodic orbits and the output function $h$ is satisfied~\cite[Ass.~5]{kohler2018mpc}, then  the optimal reachable periodic trajectory is (uniformly) exponentially stable for the resulting closed-loop system. 
Thus, the terminal ingredients enable us to implement a nonlinear version of the tracking scheme in~\cite{limon2008mpc,limon2016mpc}, that ensures exponential stability of the optimal (periodic) operation. 
More details on the theoretical properties and numerical examples can be found in~\cite{kohler2018mpc}. 
Although the consideration of general non-periodic trajectories is still an open issue, we conjecture that the approach can be extended to any class of finitely parameterized reference trajectories.

%!TEX root = ./Tracking_Journal.tex
%%%%%%%%%%%%%%%%%%%%%%%%%%%%%%%%%%%%%%%%%%%%%%%%%%%%%%%%%%%%%%%%%%%%%%%%%%%%%%%
\section{Numerical examples}
\label{sec:num}
The following examples show the applicability of the proposed method to nonlinear systems and the closed-loop performance improvement when including suitable terminal ingredients. 
We first illustrate the basic procedure at the example of a periodic reference tracking task for a continuous stirred-tank reactor (CSTR). 
Then we demonstrate the advantages of using suitable terminal ingredients with (robust) trajectory tracking and an evasive maneuver test for a car. 
Additional examples, including tracking of periodic output signals (Sec.~\ref{sec:gen_output}) with a nonlinear ball and plate system can be found in~\cite{kohler2018mpc}. 

In the following examples, the offline computation is done with an Intel Core i7 using the semidefinite programming (SDP) solver SeDuMi-1.3~\cite{sturm1999using} and the online optimization is done with CasADi~\cite{andersson2019casadi}.  
The offline computation can be done using both the discrete-time formulation (Sec.~\ref{sec:loc_stab}) or the continuous-time formulation (Appendix~\ref{sec:app_cont}). 
Hence, we also compare the performance of these different formulations.
%!TEX root = ./Tracking_Journal.tex
%%%%%%%%%%%%%%%%%%%%%%%%%%%%%%%%%%%%%%%%%%%%%%%%%%%%%%%%%%%%%%%%%%%%%%%%%%%%%%%
\subsection{Periodic reference tracking - CSTR}
\subsubsection*{System model}
We consider a continuous-time model of a continuous stirred-tank reactor (CSTR)
\begin{align*}
\begin{pmatrix}
\dot{x}_1\\\dot{x}_2\\\dot{x}_3
\end{pmatrix}
=&\begin{pmatrix}
1-x_1-10^4x_1^2\exp(\frac{-1}{x_3})-400x_1\exp(\frac{-0.55}{x_3})\\
10^4 x_1^2\exp(\frac{-1}{x_3})-x_2\\
u-x_3
\end{pmatrix},
\end{align*}
where $x_1,~x_2,~x_3$ correspond to the concentration of the reaction, the desired product, waste product and $u$ is related to the heat flux through the cooling jacket, compare~\cite{bailey1971cyclic}, \cite[Sec.~3.4]{faulwasser2018economic}. 
The constraints are 
\begin{align*}
\mathcal{Z}_r=&[0.05,0.45]\times[0.05,0.15]\times [0.05,0.2]\times[0.059,0.439],\\
\mathcal{Z}=&[0,1]^3\times[0.049,0.449].
\end{align*}
The discrete-time model is defined with explicit Runge-Kutta discretization of order $4$ and a sampling time\footnote{
In~\cite[Sec.~3.4]{faulwasser2018economic} a sampling time of $h=0.1$ is used.  
However, with the considered fourth order explicit Runge-Kutta discretization, a sampling time of $h=0.1$ does not preserve stability of the continuous-time system. 
} of $h=0.01$. 

For this system, periodic operation is economically beneficial, compare~\cite{bailey1971cyclic}. 
Thus, we consider the problem of tracking reachable periodic reference trajectories $r$ (Assumption~\ref{ass:ref}), corresponding to the economic operation of the plant. 

\subsubsection*{Offline computations}
In the following, we illustrate the reference generic offline computation for this system. 
We consider the standard quadratic tracking stage cost with $Q=I_3$, $R=10$ and use $\epsilon=0.1$.

For the continuous-time system, the Jacobian~\eqref{eq:A_r} contains four nonlinear terms, yielding the parameters
\begin{align*}
\theta_1(x)=&400\exp(-0.55/x_3),\quad 
\theta_2(x)=2\cdot 10^4 x_1\exp(-1/x_3),\\
\theta_3(x)=&10^4({x_1}/{x_3})^2\exp(-1/x_3),\\
\theta_4(x)=&400\cdot 0.55{x_1}/({x_3^2})\exp(-0.55/x_3). 
\end{align*}
The input $u_r$ enters the LMIs affinely. 
Thus, we only consider the two vertices of $u_r$ and grid $(x_1,x_3)$ using $10^2$ points. 
 
For the discrete-time system, the explicit description of the nonlinear dynamics $f$ and the corresponding Jacobian $A(r),~B(r)$ is complex. 
Thus, we directly define the non-constant\footnote{
The derivatives $\partial f_3/\partial r$, $\partial f_1/\partial x_2$, and $\partial f_2/\partial x_2$ are constant.   
} components of the Jacobian $A,~B$ as the parameters $\theta\in\mathbb{R}^6$. 
We compute the hyperbox sets $\Theta,\Omega\subseteq\mathbb{R}^6$ satisfying~\eqref{eq:Theta_set} numerically. 
For the discrete-time convex approach the  polytopic description  $\overline{\Theta}$~\eqref{eq:overline_theta} and the hyperbox description ${\Theta}\times \Omega$ (Remark~\ref{remark:box}) are considered. 
For the gridding, $(x_1,x_3,u_r,u_r^+)\in\mathbb{R}^4$ is gridded using $10^4$ points, of which approximately $8.000$ satisfy the conditions~\eqref{eq:grid_r} and are considered in the optimization problem~\eqref{eq:LMI}. 

The computational demand and the performance of the different methods are detailed in Table~\ref{tab:CSTR_comp}. 
As expected, the gridding approach yields the smallest and least conservative terminal cost. 
For this example, the convex discrete-time approach (Prop.~\ref{prop:LMI_lpv}) seems less favorable, which is mainly due to the simple description of the parameters. 
Due to the small sampling time $h$ and correspondingly small set $\Omega$, the more detailed description $\overline{\Theta}$ only marginally improves the performance but significantly increases the offline computational demand.  
Furthermore, the continuous-time formulation can be computed more efficiently. 
We note, that the parameters $Q,~R,~h$, are chosen, such that the continuous-time control law is also stabilizing for the discrete-time implementation. 
In particular, if $R$ is decreased or $h$ increased, the terminal ingredients based on the continuous-time formulation do not satisfy Assumption~\ref{ass:term} with a piece wise constant input.  
Such considerations are not necessary for the discrete-time formulation, compare Remark~\ref{rk:sample_hold} in Appendix~\ref{sec:app_cont}.

\begin{table*}
\small{
\begin{tabular}{c|c|c|c|c|cc}
\multirow{2}{*}{Method}
& \multicolumn{2}{c}{Continuous-Time}&\multicolumn{3}{c}{Discrete-time}\\
&Gridding (Lemma.~\ref{lemma:LMI_cont})&{Convex (Prop.~\ref{prop:LMI_lpv_cont})}&Gridding  (Lemma.~\ref{lemma:LMI})&
\multicolumn{2}{c}{Convex (Prop.~\ref{prop:LMI_lpv})}\\\hline
$\#$LMIs&$200$& $4^4=256$ &$7994$&$4^6=4096$&$6^6= 46.656$\\
computational time &12~s&10~s&783~s&356~s&3h~18min\\
\small{$\max_{r\in\mathcal{Z}_r}$$\lambda_{\max}(P_f(r))$}&$3.3\cdot10^3\cdot h$ &$8.3\cdot 10^4\cdot h$&$3.5\cdot 10^3$& $4\cdot 10^7$&$3.8\cdot 10^7$\\
\end{tabular}
\caption{Computational demand and conservatism of different offline computations - CSTR. }
\label{tab:CSTR_comp}
}
\end{table*}
In the following, we consider the discrete-time terminal ingredients based on Lemma~\ref{lemma:LMI}. 
Computing $\alpha_2=0.02$ using~\eqref{eq:alpha_2_better} requires $30$~s. 
Executing Algorithm~\ref{alg:offline_alpha} to ensure that $\alpha=0.02$ is valid takes $10$~min using  $2\cdot 20^4\cdot 100=3.2\cdot 10^7$ samples. 

In Figure~\ref{fig:CSTR_periodic} we can see an exemplary periodic trajectory and the corresponding terminal set\footnote{%
If $\alpha$ would be recomputed for the specific trajectory $r$, we would get $\alpha=0.1$. 
This conservatism is a result of the fact, that the previously computed value $\alpha$ needs to be valid for \textit{every} reachable reference trajectory (Ass.~\ref{ass:ref}). 
}. 
The period length is $T=1144$, which corresponds to $11.44~s$, compare~\cite[Sec.~3.4]{faulwasser2018economic}.
\begin{figure}[hbtp]
\begin{center}
\includegraphics[scale=0.5]{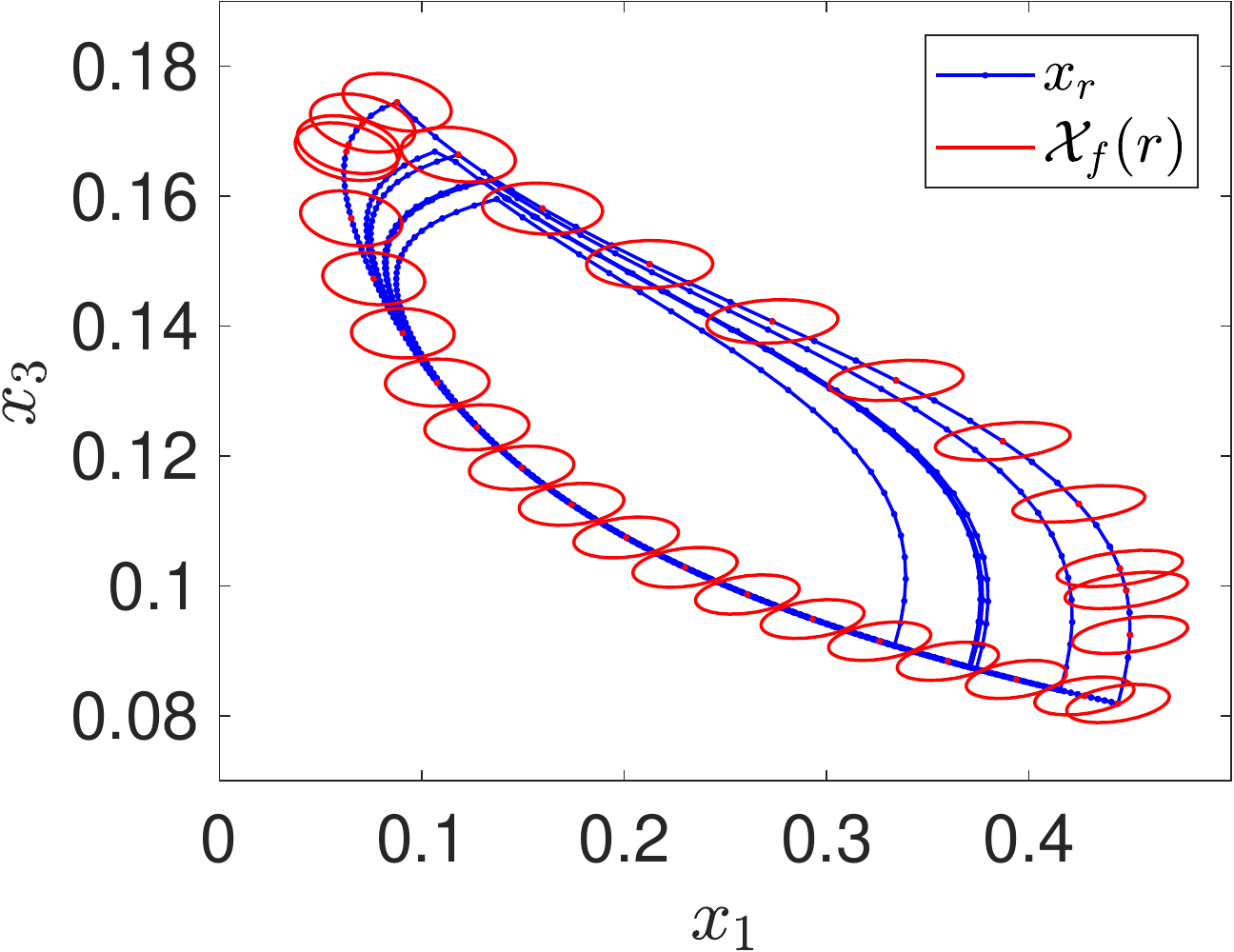}
\end{center}
\caption{Periodic trajectory - CSTR: Reference trajectory $r$ (blue) with terminal sets $\mathcal{X}_f(r)$ (red ellipses). }
\label{fig:CSTR_periodic}
\end{figure}

We wish to emphasize that this offline computation is only done once and requires no explicit knowledge of the specific trajectory or its period length $T$.
This is in contrast to the existing methods, such as~\cite{faulwasser2011model,aydiner2016periodic} which would compute terminal ingredients for a specific reference trajectory and thus could not deal with online changing operation conditions (e.g. due to changes in the price signal~\cite{ferramosca2014economic}).

%!TEX root = ./Tracking_Journal.tex
%%%%%%%%%%%%%%%%%%%%%%%%%%%%%%%%%%%%%%%%%%%%%%%%%%%%%%%%%%%%%%%%%%%%%%%%%%%%%%%
\subsection{Automated driving - robust reference tracking}  
The following example shows the applicability of the proposed procedure to nonlinear robust reference tracking and demonstrates the performance improvement of including suitable terminal ingredients. 
\subsubsection*{System model} 
We consider a nonlinear kinematic bicycle model of a car 
\begin{align*}
\dot{z}_1=&v\cos(\psi+\beta),\quad 
\dot{z}_2=v\sin(\psi+\beta),\\
\dot{\psi}=&v/l_r\sin(\beta),\quad 
\dot{v}=a,\quad 
\dot{\delta}=u_{\delta},\\
\beta=&\tan^{-1}\left(\dfrac{l_r}{l_f+l_r}\tan(\delta)\right),\\
x=&[z_1,z_2,\psi,v,\delta]^{\top}\in\mathbb{R}^5,\quad u=[a,u_{\delta}]^{\top}\in\mathbb{R}^2,
\end{align*}
with the position $z_i$, the inertial heading $\psi$, the velocity $v$, the front steering angle $\delta$, the acceleration $a$ and the change in the steering angle $u_{\delta}$.  
The model constants $l_f=1.4$ and $l_r=1.5$ represent the distance of the center of mass to the front and rear axle. 
More details on kinematic bicycle models can be found in~\cite{kong2015kinematic}. 
The (non-compact) constraint sets are given by 
\begin{align*}
\mathcal{Z}_r=&\{v\in[10,50],a\in[-1,1],\delta\in[-0.4,0.4],u_{\delta}\in[-3,3]\},\\
\mathcal{Z}=&{\{v\in[5~,55],a\in[-2,2],\delta\in[-0.5,0.5],u_{\delta}\in[-6,6]\}}.
\end{align*}	 

\subsubsection*{Offline computations}
We consider the stage cost $Q=I_5$, $R= I_2$ and $\epsilon=0.1$ and use an Euler discretization with the step size $h=2ms$. 
Computing the linearization~\eqref{eq:A_r} and using a quasi-LPV parameterization~\eqref{eq:A_paramlin} results in $\theta\in\mathbb{R}^8$, where the parameters $\theta$ consist of trigonometric functions in $\Psi,~\delta$ and are linear in the velocity $v$.  

For this example, the convex approach (Prop.~\ref{prop:LMI_lpv}) is not feasible, since the simple and conservative hyperbox\footnote{%
This description does not take into account that $\sin(\psi+\beta)$ and $\cos(\psi+\beta)$ cannot be zero simultaneously. 
This issue can be circumvented by considering a more detailed description of $\Theta$, e.g. using coupled ellipsoidal constraints. } description $\theta\in\Theta$ includes linearized dynamics which are not stabilizable. 

For the gridding, we consider both the discrete-time and a continuous-time formulation (compare Appendix~\ref{sec:app_cont}). 
In the continuous-time formulation $a$ and $u_{\delta}$ enter the LMIs affinely. 
Thus, we only consider the $2^2=4$ vertices of $(a,u_{\delta})$ and grid $(\psi,v,\delta)$ using $10^3$ points. 
For the discrete-time formulation~\eqref{eq:LMI} the LMIs are not affine in $u_{\delta}$ and thus we grid $(\psi,v,\delta,u_{\delta})$ using $10^3\cdot 5$ points and consider the two vertices of $a$.    
The dimensions of the corresponding LMI-blocks are $(2n+m)\times (2n+m)=12\times 12$ and $(3n+m)\times (3n+m)=17\times 17$ , respectively.
The following table captures the weighting of the terminal cost and the computational effort of the proposed approach. 
\begin{tabular}{c|c|c}
\multirow{2}{*}{Method}
& {Continuous-Time}&{Discrete-time}\\
& (Lemma.~\ref{lemma:LMI_cont})&  (Lemma.~\ref{lemma:LMI}) \\\hline
$\#$LMIs-blocks&$10^3\cdot 2^2=4\cdot 10^3$  &$10^3\cdot 5\cdot 2=10^4$ \\\
comp.~time&14~min&33~min\\
\small{$\max_{r\in\mathcal{Z}_r}$$\lambda_{\max}(P_f(r))$}&$8.0\cdot 10^4\cdot h$&$8.4\cdot 10^4$
\end{tabular}  
\begin{remark}
\label{rk:cont_discrete}
For the considered example and parameters, the continuous-time terminal cost is also valid for a zero-order hold discrete-time implementation with $h=2~ms$. 
This is in general not the case. 
For example if $R=10^{-4}$ or $h=10~ms$ is chosen, the terminal ingredients based on the continuous-time offline optimization are not stabilizing for the discrete-time system. 
If the continuous-time offline procedure is used, the computation (and thus verification) of $\alpha$ for the discrete-time system using Algorithm~\ref{alg:offline_alpha} is crucial. 
This issue is also discussed in Remark~\ref{rk:sample_hold} of Appendix~\ref{sec:app_cont}. 
\end{remark}
In the following, we only consider the discrete-time terminal ingredients based on Lemma~\ref{lemma:LMI}. 
Executing Algorithm~\ref{alg:offline_alpha} to ensure that $\alpha_1=10^4$ is valid takes $25$~min using  $20^3\cdot 10^2\cdot 100= 8\cdot 10^7$ samples. 
 
\subsubsection*{Robust trajectory tracking - Evasive maneuver test}
In order to demonstrate the applicability of the proposed tracking MPC scheme, we consider an evasive maneuver test  (compare~ISO norm 3888-2~\cite{ISO3888}).  
In this scenario a car is driving with $v=20~m/s$ and performs two consecutive lane changes to simulate the avoidance of a possible obstacle. 
The basic setup, with a feasible reference trajectory $r$, additional path constraints\footnote{%
Ideally, these constraints should restrict the overall position of the vehicle. 
For simplicity we treat them as (time-varying) polytopic constraints on $z_2$, that require the $z_2$ position to be within a margin of $\pm 35~cm$. 
} $\mathcal{X}$ and the terminal set (projected on $z_1\times z_2$) can be seen in Figure~\ref{fig:elchtest}. The terminal set size is restricted by the input constraint on $u_\delta$ and the path constraint $\mathcal{X}$, yielding the terminal set size $\alpha=\alpha_2\approx 10^2$. 
For comparison, we also computed a terminal cost for this specific \textit{given} trajectory based on an LTV description~\cite{faulwasser2011model}. 
The generic offline computation results in a roughly five times larger terminal cost, which gives an indication of the conservatism.  

\begin{figure}[hbtp]
\begin{center}
\includegraphics[width=0.4\textwidth]{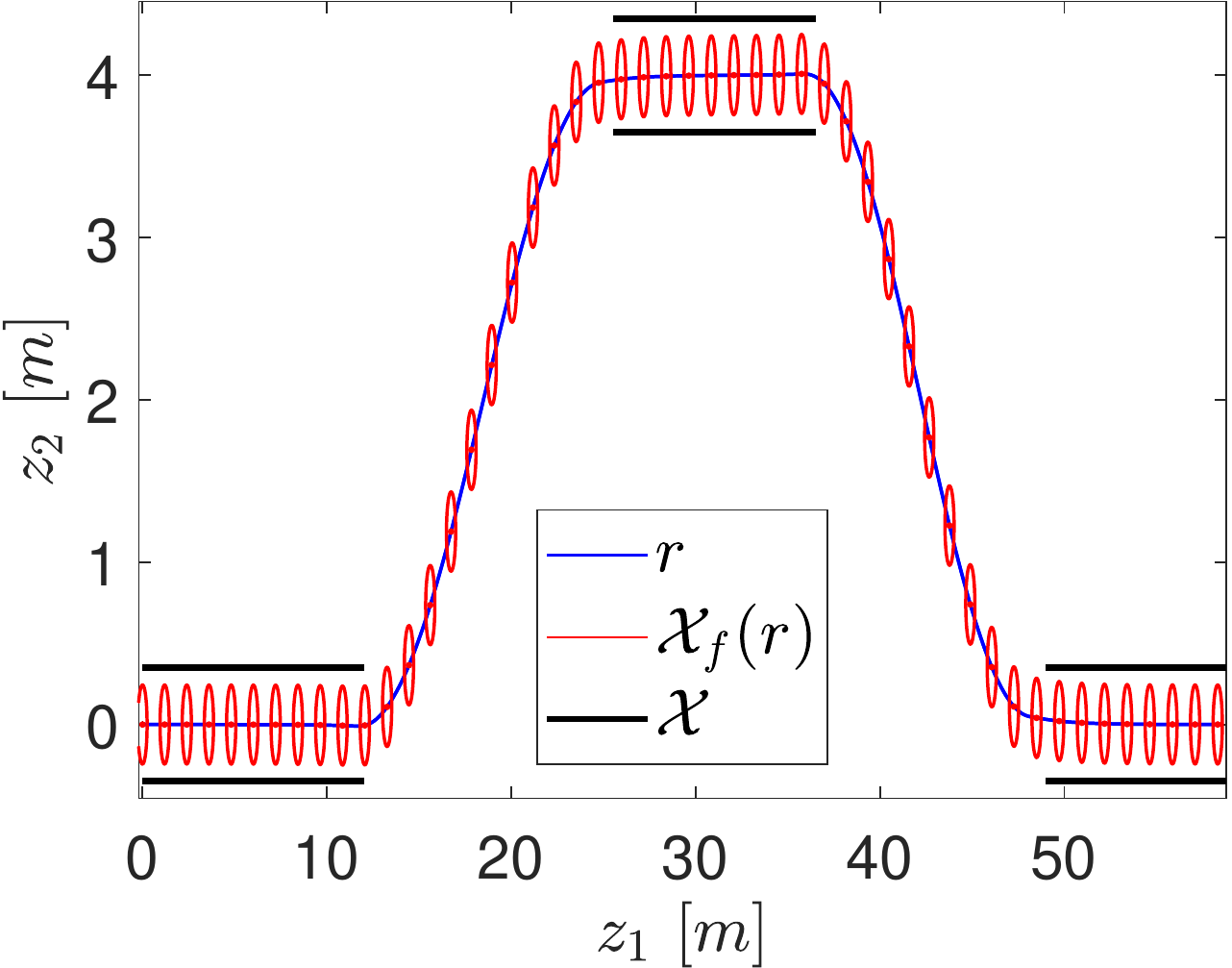}
\end{center}
\caption{ Evasive maneuver test: Reference trajectory $r$ (blue), terminal sets $\mathcal{X}_f(r)$ (red) and additional state constraints $\mathcal{X}$ (black).}
\label{fig:elchtest}
\end{figure}

In order to show that the proposed approach can be applied under realistic conditions, we consider additive disturbances $w(t)\in\mathbb{R}^n$ and a prediction horizon of $N=10$.  
To ensure robust constraint satisfaction, we use the constraint tightening method proposed in~\cite{kohler2018novel}, which is based on the achievable contraction\footnote{%
This property is verified by computing a terminal cost, which is valid on the full constraint set $\mathcal{Z}$, compare Prop.~\ref{prop:increm} and App.~\ref{sec:app_robust}. 
Analogous to the computation of $\alpha$, the numerical value of $\rho$ can be ascertained using Alg.~\ref{alg:offline_alpha}. 
} rate $\rho=0.9995$. 
To ensure robust recursive feasibility, the terminal set needs to be robust positively invariant, which can be ensured for $\|w(t)\|\leq \hat{w}=1.82\cdot 10^{-5}=9.1\cdot 10^{-3}h$, compare \eqref{eq:w_2} in Proposition~\ref{prop:RPI} of Appendix~\ref{sec:app_robust}. 
The constraints are tightened over the prediction horizon with a scalar using the method in~\cite{kohler2018novel}
\begin{align*}
(x(k|t),u(k|t))\in(1-\epsilon_{k})\mathcal{Z},
\quad \epsilon_k=\epsilon\frac{1-{\rho}^k}{1-{\rho}},
\end{align*}
with $\epsilon=2.5\cdot 10^{-4}$.
The resulting robust tracking MPC scheme guarantees (uniform) practical exponential stability and robust constraint satisfaction, for details see Appendix~\ref{sec:app_robust} and~\cite{kohler2018novel}.

We simulated the closed-loop MPC using random disturbances $\|w(t)\|=\hat{w}$ and compared the performance to MPC without terminal constraints ($V_f=0$, UC,~\cite{kohlernonlinear19}) and MPC with terminal equality constraint ($\mathcal{X}_f(r)=x_r$, TEC). 
To enable a comparison of the computational demand we fixed the number of iterations in CasADi to $1$ per time step, resulting in online computation time of approx $13$~ms  for all three approaches. 
The corresponding results can be seen in Figures~\ref{fig:Car_closedloop_1} and~\ref{fig:Car_closedloop_2}. 
\begin{figure}[hbtp]
\begin{center}
\includegraphics[width=0.5\textwidth]{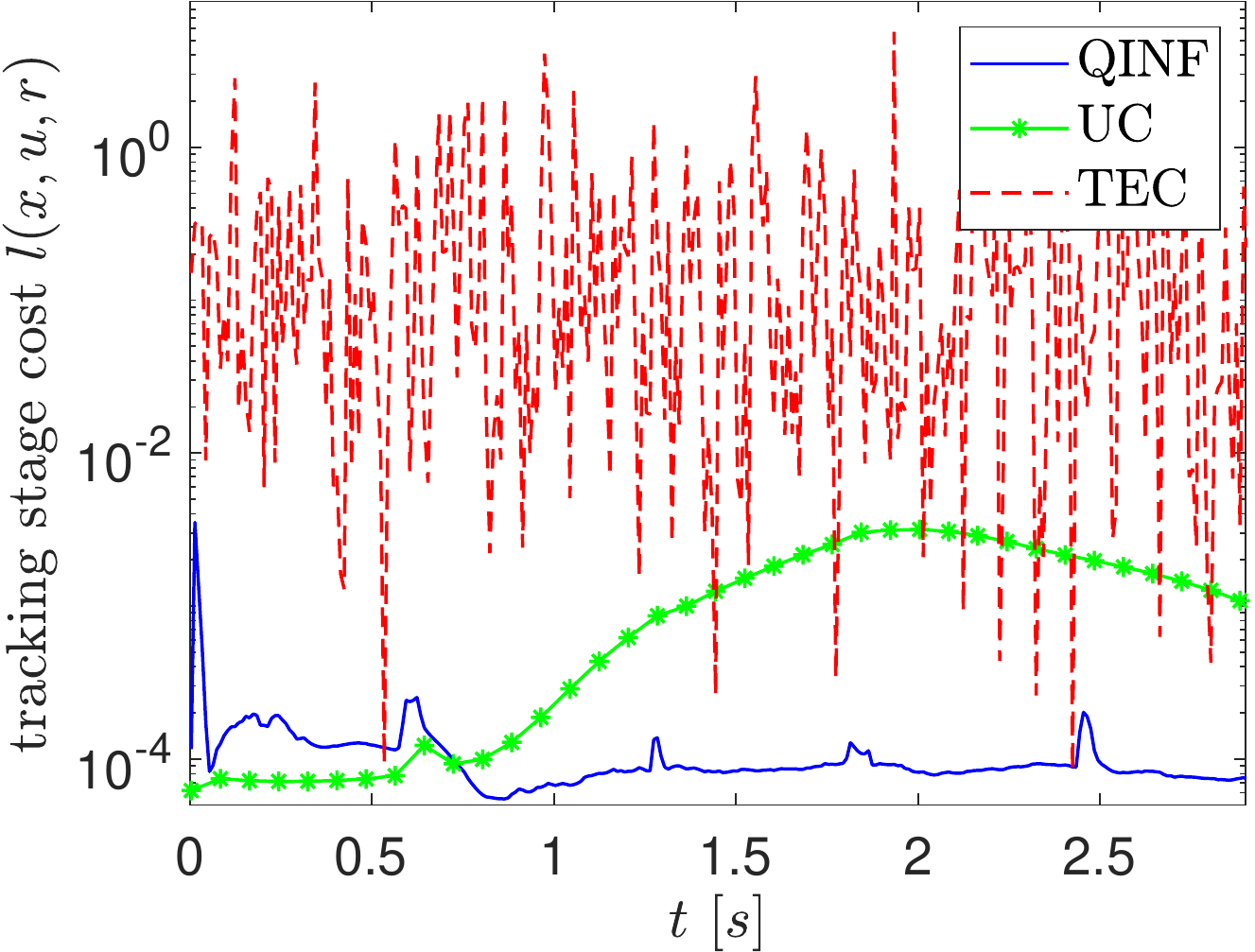}
\end{center}
\caption{ Evasive maneuver test: Closed-loop tracking stage cost for the proposed terminal constraint tracking MPC (blue,solid,QINF), a corresponding tracking MPC scheme without terminal constraints (green,dash-star,UC) and an MPC scheme with a terminal equality constraint (red,dashed,TEC)}
\label{fig:Car_closedloop_1}
\end{figure}
\begin{figure}[hbtp]
\begin{center}
\includegraphics[width=0.4\textwidth]{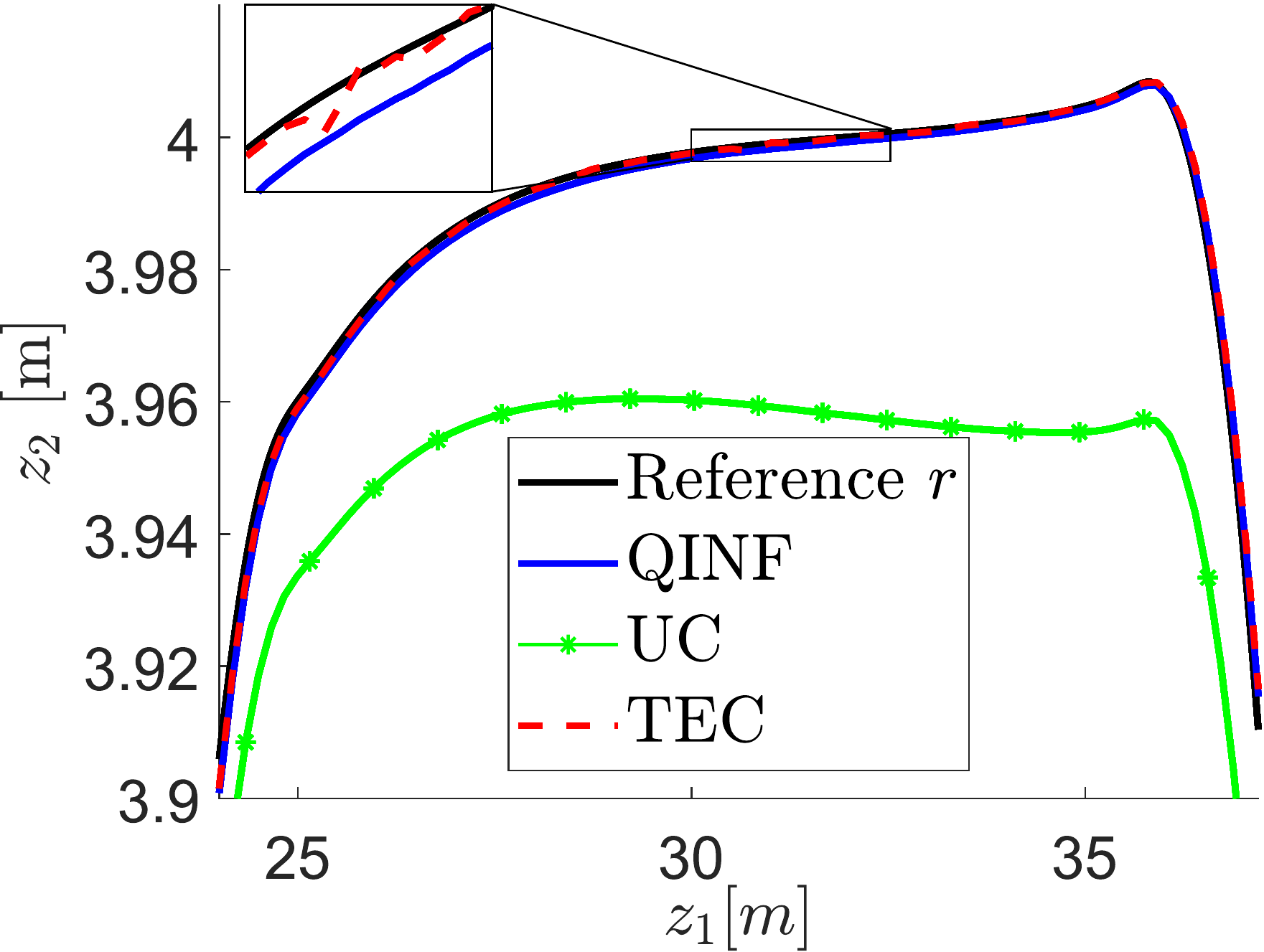}
\end{center}
\caption{Evasive maneuver test: Closed-loop trajectory of $z_1,~z_2$ over the time interval $t\in[1.32s,1.81s]$ with the reference $r$ (black,solid), the MPC based on the proposed terminal ingredients (blue,solid,QINF), a corresponding tracking MPC scheme without terminal constraints (green,dash-star,UC) and an MPC scheme with a terminal equality constraint (red,dashed,TEC). }
\label{fig:Car_closedloop_2}
\end{figure}

The closed-loop performance (as measured by the tracking stage cost
\footnote{%
If we ignore the input tracking stage cost and only consider $\|x-x_r\|_Q^2$ as the performance, then the TEC has only $13\%$ of the tracking error of QINF and UC has $30$-times the tracking error. 
If, for some reason, we would only be interested in the tracking error in the input $\|u-u_r\|_R^2$, then UC has only $48\%$ of the error of QINF and TEC has $4.5\cdot 10^3$ times the error of QINF.
}) of UC and TEC are $10$ and $3.000$ times larger than the proposed scheme with the terminal cost (QINF), compare Figure~\ref{fig:Car_closedloop_1}.  
Specifically, the MPC without terminal constraints (UC) has a significant (growing) tracking error in the position (see Figure~\ref{fig:Car_closedloop_2}), since the UC with a short horizon typically leads to a slower convergence with smaller control action (as stability is not explicitly enforced). 
On the other side, the terminal equality constraint MPC (TEC) has large deadbeat like input oscillations, which is a result of the terminal constraint with the short prediction horizon.  
UC and TEC achieve a similar performance to QINF with $N=10$, if the prediction horizon\footnote{%
For this second comparison, we did not limit the number of iterations for UC and TEC, since we were unable to achieve a similar performance with UC using only $1$ iterations (which may be due to the lack of a good warmstart).
} is increased to $N=23$ and $N=59$, respectively. 
This increases the online computational demand compared to QINF by $100\%$  and $300\%$, respectively.

The proposed MPC scheme robustly achieves a small tracking error with a short prediction horizon. 
This shows that including (suitable) terminal ingredients significantly reduces the tracking error and improves the closed-loop performance, as also articulated in~\cite{mayne2013apologia}. 

%!TEX root = ./Tracking_Journal.tex
%%%%%%%%%%%%%%%%%%%%%%%%%%%%%%%%%%%%%%%%%%%%%%%%%%%%%%%%%%%%%%%%%%%%%%%%%%%%%%%
\section{Conclusion}
\label{sec:sum}
We have presented a procedure to compute terminal ingredients for nonlinear reference tracking MPC schemes offline. 
The main novelty in this approach is that the offline computation only needs to be done once, irrespective of the setpoint or trajectory to be stabilized. 
This is possible by computing parameterized terminal ingredients and approximating the nonlinear system locally as a quasi-LPV system, with the reference trajectory to be stabilized as the parameter.  
Furthermore, we have shown that the reference generic offline computation enables us to design nonlinear MPC schemes that ensure optimal periodic operation   despite online changing operation conditions. 
We have demonstrated the applicability and advantages of the proposed procedure with numerical examples.

The extension of the proposed procedure to large scale nonlinear distributed systems using a seperable formulation is part of future work.

\bibliographystyle{IEEEtran}  
\bibliography{Literature_short}  
\clearpage 
%!TEX root = ./Tracking_Journal.tex
%%%%%%%%%%%%%%%%%%%%%%%%%%%%%%%%%%%%%%%%%%%%%%%%%%%%%%%%%%%%%%%%%%%%%%%%%%%%%%%
\appendix
In Appendix \ref{sec:app_increm}, the connection between \textit{incremental system properties} and the considered reference generic terminal ingredients are discussed. 
In Appendix~\ref{sec:app_robust}, these incremental stability properties are used to extend the approach to \textit{robust} reference tracking, by introducing a simple constraint tightening to ensure robust constraint satisfaction under additive disturbances.  
In Appendix~\ref{sec:app_cont}, the derivations for the reference generic offline computations (Prop.~\ref{prop:LMI_lpv}) are extended to \textit{continuous-time} systems.
In Appendix~\ref{sec:app_output}, the procedure is extended to nonlinear  \textit{output tracking stage costs}, for both discrete-time and continuous-time systems.

 \subsection{(Local) Incremental exponential stabilizability}
\label{sec:app_increm}
In the following we clarify the connection between incremental stabilizability properties and the terminal ingredients.
\begin{definition}
\label{def:increm_stab}
A set of reference trajectories $r$ specified by some dynamic inclusion $r(t+1)\in\mathcal{R}(r(t))$ is locally incrementally exponentially stabilizable for the system~\eqref{eq:sys},  if there exist constants $\rho\in(0,1),M,c>0$ and a control law $\kappa(x,r)$, such that for any initial condition
 satisfying $\|x(0)-x_r(0)\|\leq c$, the trajectory $x(t)$ with $x(t+1)=f(x(t),\kappa(x(t),r(t)))$ satisfies $\|x(t)-x_r(t)\|\leq M\rho^t\|x(0)-x_r(0)\|$,~$\forall t\geq 0$. 
\end{definition}
This definition is closely related to the concept of universal exponential stabilizability~\cite{manchester2017control}, which characterizes the stabilizability of arbitrary trajectories in continuous-time. 
One of the core differences in the definitions is the treatment of constraints, i.e. we study stabilizability of classes of trajectories $r$ that satisfy certain constraints, compare Assumption~\ref{ass:ref} and Remark~\ref{rk:ref}. 
This difference is crucial when discussing local versus global stabilizability and constrained control. 

The following proposition shows that the conditions in Lemma~\ref{lemma:lpv} directly imply local incremental exponential stabilizability of the reference trajectory. 
\begin{proposition}
\label{prop:increm}
Suppose that there exist matrices $P_f(r),~K_f(r)$ that satisfy the conditions in Lemma~\ref{lemma:lpv}. 
Then the control law $k_f(x,r)=u_r+K_f(r)(x-x_r)$ locally incrementally exponentially stabilizes any reference $r$ satisfying Assumption~\ref{ass:ref}.
\end{proposition}
\begin{proof}
The following proof follows the arguments of~\cite[Prop.~1,2]{kohlernonlinear19}. 
For any $\|x(0)-x_r(0)\|\leq c$ with $c=\sqrt{\alpha/c_{u}}$, we have $x(0)\in\mathcal{X}_f(r)$, with $\alpha,~c_{u}$ according to Lemma~\ref{lemma:lpv}. 
Thus, the terminal cost $V_f(x,r)$ is a local incremental Lyapunov function that satisfies 
\begin{align*}
V_f(x(t+1),r(t+1))\leq \rho^2 V_f(x(t),r(t)),~ \rho^2=1-\dfrac{\lambda_{\min}(Q)}{c_u},
\end{align*}
and thus 
\begin{align*}
\|x(t)-x_r(t)\|\leq {\rho}^tM \|x_r(0)-x(0)\|,\quad M=\sqrt{{c_u}/{c_l}}.
\end{align*}
\end{proof}
\begin{remark}
\label{rk:increm}
This result establishes local incremental stabilizability with the incremental Lyapunov function $V_f(x,r)$ based on properties of the linearization, compare~\cite[Prop.~1]{kohlernonlinear19}. 
This system property is a natural extension of previous works on incremental stability and corresponding incremental Lyapunov functions, see~\cite{angeli2002lyapunov,tran2016incremental},~\cite[Ass.~1]{kohlernonlinear19}.  
This property implies stabilizability of $(A(r),~B(r))$ around any (fixed) steady-state $r^+=r$, but it does not necessarily imply stabilizability of $(A(r),~B(r))$ for arbitrary $r\in\mathcal{Z}_r$, as $P_f(r)$ might decrease along the trajectory. 

For continuous-time systems, an analogous result exists based on contraction metrics and universal stabilizability~\cite{manchester2017control}. 
\end{remark}
The following proposition shows that in the absence of constraints we recover non-local results similar to~\cite{manchester2017control}. 
\begin{proposition}
\label{prop:univ}
Consider $\mathcal{Z}_r=\mathcal{Z}=\mathbb{R}^{n+m}$. 
Suppose that there exist matices $P_f(r),~K_f(r)$ that satisfy the conditions in Lemma~\ref{lemma:lpv}. 
Assume further that $c_lI\leq P_f(r)\leq c_uI$ for all $r\in\mathbb{R}^{n+m}$ with some constants $c_l,~c_u$ and $K_f(r)=K(x_r)$. 
Then any  reference $r$ satisfying Assumption~\ref{ass:ref} is exponentially incrementally stabilizable with the control law 
\begin{align*}
\kappa(x,r)=K(x)x-K(x_r)x_r+u_r,
\end{align*} 
i.e., for any initial condition $x(0)\in\mathbb{R}^n$ the state trajectory $x(t+1)=f(x(t),\kappa(x(t),r(t)))$ satisfies $\|x(t)-x_r(t)\|\leq M\rho^t\|x(0)-x_r(0)\|$. 
\end{proposition}
\begin{proof}
Consider an auxiliary (pre-stabilized) system defined by $\tilde{f}(x,v)=f(x,K(x)x+v)$. 
Consider a reference $r$ generated by some input trajectory $u_r$ with the system dynamics~\eqref{eq:sys} (Ass.~\ref{ass:ref}) and some initial condition $x_r(0)$ resulting in the state reference $x_r$. 
Now, consider a reference $\tilde{r}$ generated by the input $v_r(t)=u_r(t)-K(x_r(t))x_r(t)$ with the system dynamics according to $\tilde{f}$ and the same initial condition. 
Due to the definition of the auxiliary system we have  $\tilde{x}_r(t)=x_r(t)$, $\forall t\geq 0$. 
For an arbitrary, but fixed input $v$, stability of the reference trajectory is equivalent to contractivity of the nonlinear time-varying system $\tilde{f}(x,t)$. 
This can be established with the contractivity metric $P_f(r(t))=P_f(x_r(t),t)$, compare~\cite{lohmiller1998contraction}. 
\end{proof}
In the absence of constraints, it is crucial that $P_f$ has a constant lower and upper bound.  
If the matrix $K_f$ depends on the full reference $r$ (not just $x_r$), the controller $\kappa$ in Proposition~\ref{prop:univ} is not necessarily well defined. 
\begin{remark}
\label{rk:track_vs_stab}
The relation between the controller $k_f$ (Prop.~\ref{prop:increm}) and $\kappa$ (Prop.~\ref{prop:univ}), is that of reference tracking versus pre-stabilization. 
The first one is more natural in the context of tracking MPC and contains existing results for the design of terminal ingredients as special cases~\cite{chen1998quasi,faulwasser2011model,aydiner2016periodic}. 
The second controller $\kappa$ allows for non-local stability results and is more suited for unconstrained control problems~\cite{manchester2017control}. 
For constant matrices $K$ the two controllers are equivalent, but the incremental Lyapunov functions (and thus terminal costs) are differently parameterized ($P_f(r)$, $P_f(x,u)$).  
\end{remark}
\begin{remark}
\label{rk:term_other}
The problem of computing reference generic terminal ingredients is equivalent to computing an incrementally stabilizing controller and is thus strongly related to the computation of robust positive invariant (RPI) tubes in nonlinear robust MPC schemes, compare~\cite{bayer2013discrete,kohler2018novel}. 
For comparison, in~\cite{yu2010robust,yu2013tube} constant matrices $P_f,~K_f$ are computed that certify incremental stability for continuous-time systems (by considering small Lipschitz nonlinearities or by describing the linearization as a convex combination of different linear systems). 
This approach can be directly extended to more general nonlinear systems using the proposed terminal ingredients. 
In particular, by changing the stage cost to 
\begin{align*}
\ell(x,u,r)=\|u-u_r+K(x_r)x_r-K(x)x\|_R^2  
\end{align*}
 one can design a nonlinear version of~\cite{chisci2001systems}, compare also~\cite{kohler2018novel}. 
A detailed description of a corresponding nonlinear robust tube based (tracking) MPC scheme based on incremental stabilizability can be found in Appendix~\ref{sec:app_robust}. 
\end{remark}
\begin{remark}
\label{rk:non_quadratic}
In case a system is not exponentially stabilizable~\cite{muller2017quadratic}, it might be possible to make a nonlinear transformation resulting in a quadratically stabilizable system,  (see for example nonlinear systems in normal form~\cite{isidori2013nonlinear}).  
\end{remark}

\subsection{Robust reference tracking} 
\label{sec:app_robust}
In the following, we summarize the theoretical results for robust reference tracking based on the reference generic terminal ingredients and~\cite{kohler2018novel}, where robust setpoint stabilization without terminal constraints was considered. 
This method is applicable to nonlinear incrementally stabilizable systems (Sec.~\ref{sec:app_increm}) with polytopic constraints and additive disturbances and can be thought of as a nonlinear version of~\cite{chisci2001systems}. 
%-
\subsubsection{Setup}
We consider nonlinear discrete-time systems subject to additive bounded disturbances and polytopic constraints
\begin{align*}
x(t+1)=&f(x(t),u(t))+w(t),\quad \|w\|\leq \hat{w},\\ 
\mathcal{Z}=&\{r\in\mathbb{R}^{n+m}|~L_j r\leq 1,\quad j=1,\dots,q\}. 
\end{align*}
\subsubsection{Incremental stabilizability}
\begin{assumption}
\label{ass:contract} \cite[Ass.~1]{kohler2018novel}\cite[Ass.~1]{kohlernonlinear19}
There exist a control law $\kappa:\mathbb{R}^n\times\mathcal{Z}\rightarrow\mathbb{R}^m$, an incremental Lyapunov function $V_{\delta}:\mathbb{R}^n\times \mathcal{Z}\rightarrow\mathbb{R}_{\geq 0}$, that is continuous in the first argument and satisfies $V_{\delta}(x_r,x_r,u_r)=0$ for all $(x_r,u_r)\in\mathcal{Z}$, and parameters $c_{\delta,l},~c_{\delta,u},~\delta_{\text{loc}},~c_{j}\in\mathbb{R}_{>0}$, $\rho\in(0,1)$, such that the following properties hold for all $(x,x_r,u_r)\in\mathbb{R}^n\times\mathcal{Z}$, $r^+=(x_r^+,u_r^+)\in\mathcal{Z}$ with $V_{\delta}(x,r)\leq \delta_{\text{loc}}$: 
\begin{subequations}
\begin{align}
c_{\delta,l}\|x-x_r\|^2\leq V_{\delta}(x,x_r,u_r)\leq& c_{\delta,u}\|x-x_r\|^2,\\
\label{eq:lipschitz}
L_j (x-x_r,\kappa(x,x_r,u_r)-u_r)\leq& c_{j}\sqrt{V_{\delta}(x,x_r,u_r)},\\
\label{eq:contract}
V_{\delta}(x^+,x_r^+,u_r^+)\leq& \rho^2 V_{\delta}(x,x_r,u_r),
\end{align}
\end{subequations}
 with $x^+=f(x,\kappa(x,x_r,u_r))$, $x_r^+=f(x_r,u_r)$, $j=1,\dots,q$. 
\end{assumption}
This assumption implies incremental stabilizability (Def.~\ref{def:increm_stab}) for all feasible trajectories $r$, i.e., $r(t+1)\in\mathcal{R}(r(t))$ (Ass.~\ref{ass:ref}). 
For $\kappa(x,x_r,u_r)=u_r$ this reduces to incremental stability and correspondingly the robust MPC method in~\cite{bayer2013discrete} can also be used. 
This assumption can be verified by using Algorithm~\ref{alg:offline} to compute a terminal cost that is valid on $\mathcal{Z}$, compare Proposition~\ref{prop:increm}. 
The contraction rate $\rho$~\eqref{eq:contract}, is used to design a generic constraint tightening to ensure robust constraint satisfaction. 
The condition~\eqref{eq:lipschitz} is satisfied if the control law $\kappa$ is locally Lipschitz continuous, compare also~\cite{kohler2018novel}. 
\subsubsection{Constraint tightening} The constraints are tightened using the following scalar operations
\begin{align*}
\epsilon_j=&c_{j}\sqrt{c_u}\hat{w},\quad 
\epsilon_{j,k}=\dfrac{1-{\rho}^k}{1-{\rho}}\epsilon_j, k=0,\dots,N,\\
\mathcal{Z}_k=&\{r\in\mathbb{R}^n|~L_j r\leq 1-\epsilon_{j,k},\quad j=1,\dots, q\}. 
\end{align*}
The following bound on the disturbance is required to ensure that the tightened constraints are non-empty, i.e., $0\in\text{int}\left(\mathcal{Z}_N\right)$:
\begin{align}
\label{eq:w_1}
\hat{w}<\dfrac{1}{\max_jc_j}\dfrac{1}{\sqrt{c_u}}\dfrac{1-\rho^N}{1-\rho},
\end{align}
\subsubsection{Terminal ingredients}
In~\cite{kohler2018novel} the robust constraint tightening is considered for an MPC scheme without terminal constraints, compare Remark~\ref{rk:withoutterm}.
Some details regarding the extension/modification of the robust MPC scheme to a setting with terminal constraints are based on~\cite{hertneck2018learning}. 
\begin{assumption}
\label{ass:term_dist} 
There exist matrices $K_f(r)\in\mathbb{R}^{m\times n}$, $P_f(r)\in\mathbb{R}^{n\times n}$ with $c_l I_n\leq P_f(r)\leq c_u I_n$, a terminal set $\mathcal{X}_f(r)=\{x\in\mathbb{R}^n|~V_f(x,r)\leq \alpha_w\}$ with the terminal cost $V_f(x,r)=\|x-x_r\|_{P_f(r)}^2$, such that the following properties hold for any $r\in\mathcal{Z}_r$, any $x\in\mathcal{X}_f(r)$, any $r^+\in\mathcal{R}(r)$ and any $w\in\mathcal{W}_N$
\begin{subequations}
\label{eq:term_dist}
\begin{align}
\label{eq:term_dist_dec}
V_f(x^+,r^+)\leq& V_f(x,r)- \ell(x,k_f(x,r),r),\\
\label{eq:term_dist_RPI}
V_f(x^++w,r^+)\leq& \alpha_w,\\
\label{eq:term_dist_con}
(x,k_f(x,r))\in&\mathcal{Z}_N,
\end{align}
\end{subequations}
with $x^+=f(x,k_f(x,r))$, $k_f(x,r)=u_r+K_f(r)\cdot (x-x_r)$, $\mathcal{W}_N=\{w\in\mathbb{R}^n|~\|w\|\leq \hat{w}_N= \hat{w} \rho^N\sqrt{c_{\delta,u}/c_{\delta,l}} \}$, and positive constants $c_l,~c_u,~\alpha_w$. 
\end{assumption}
Compared to the nominal case (Ass.~\ref{ass:term}), we have a smaller terminal set size $\alpha_w$ due to the tightened constraints~\eqref{eq:term_dist_con} and an RPI condition that needs to be verified~\eqref{eq:term_dist_RPI}. 
Due to the quadratic nature of the terminal cost and the stage cost, \eqref{eq:term_dist_dec} implies $V_f(x^+,r^+)\leq \rho_f^2V_f(x,r,)$, with some $\rho_f\in(0,1)$, e.g. $\rho_f=1-\lambda_{\min}(Q)/c_u$.
\begin{proposition}
\label{prop:RPI}
Let Assumption~\ref{ass:term} hold and assume that $\hat{w}$ satisfies~\eqref{eq:w_1}. 
Then the terminal ingredients (Ass.~\ref{ass:term}) satisfy~\eqref{eq:term_dist_con} with a positive constant $\alpha_w$.
Suppose further that
\begin{align}
\label{eq:w_2}
\hat{w}\leq\sqrt{\dfrac{\alpha_w c_{\delta,l}}{c_{\delta,u}c_{u}}}\dfrac{1-\rho_f}{\rho^N}.
\end{align}
Then~\eqref{eq:term_dist_RPI} and thus Assumption~\ref{ass:term_dist} is satisfied.   
\end{proposition}
\begin{proof}
Condition~\eqref{eq:term_dist_dec} directly follows fom Assumption~\ref{ass:term}. 
Inequality~\eqref{eq:w_1} ensures that $0\in\text{int}(\mathcal{Z}_N)$, which in combination with the quadratic bounds on $V_f$ and linear bounds on $k_f$ ensures that~\eqref{eq:term_dist_con} is satisfied for some positive constant $\alpha_w$, compare the proof of Lemma~\ref{lemma:lpv}, Algorithm~\ref{alg:offline_alpha} and the optimization problem~\eqref{eq:alpha_2_better} for the computation of $\alpha_w$.   

Using the quadratic nature of the terminal cost, a sufficient condition for~\eqref{eq:term_dist_RPI} is given by
\begin{align*}
&V_f(x^++w,r^+)\\
\leq& V_f(x^+,r^+)+2\sqrt{c_uV_f(x^+,r^+)}\hat{w}_N+c_u \hat{w}_N^2\leq \alpha_w,
\end{align*}
with $\|w\|\leq \hat{w}_N$. 
Using the contraction rate $\rho_f$ to bound $V_f(x^+,r^+)\leq \rho_f^2\alpha_w$, this condition reduces to $\hat{w}_N\leq (1-\rho_f)\sqrt{\alpha_w/c_u}$. 
The inequality on $\hat{w}$ follows from the definition of $\hat{w}_N$. 
\end{proof}
\subsubsection{Robust tracking MPC}
The robust tracking MPC is based on the following MPC optimization problem
\begin{subequations}
\label{eq:robust_MPC}
\begin{align}
V(x(t),r(\cdot|t))=\min_{u(\cdot|t)}&J_N(x(\cdot|t),u(\cdot|t),r(\cdot|t))\\
\text{s.t. }&x(k+1|t)=f(x(k|t),u(k|t)),\\
&x(0|t)=x(t),\\
&(x(k|t),u(k|t))\in\mathcal{Z}_k,\\
&x(N|t)\in\mathcal{X}_f({r}(N|t)).
\end{align}
\end{subequations}
Compared to~\eqref{eq:MPC}, in this optimization problem the state and input constraints are tightened. 
\subsubsection{Theoretical guarantees}
\begin{theorem}
\label{thm:robust}
Let Assumptions~\ref{ass:ref},~\ref{ass:contract} and~\ref{ass:term_dist}  hold. 
Assume further that $\hat{w}\leq \sqrt{\delta_{\text{loc}}/c_{\delta,u}}$ and that~\eqref{eq:robust_MPC} is feasible at $t=0$. 
The optimization problem~\eqref{eq:robust_MPC} is recursively feasible and the tracking error $e_r=0$ is (uniformly) practically exponentially stable for the resulting closed-loop system~\eqref{eq:close}. 
\end{theorem}
\begin{proof}
The proof is analogous to~\cite{kohler2018novel}, except for the satisfaction of the terminal constraint, which is guaranteed by Assumption~\ref{ass:term_dist}, compare also~\cite[Thm.~7]{hertneck2018learning}. 
\end{proof}
Note that both the size of the constraint set~\eqref{eq:w_1} and the local incremental stabilizability~\eqref{eq:w_2} lead to hard bounds on the size of the disturbance $\hat{w}$, that can be considered in this approach.  
This approach can also be extended to utilize a general nonlinear state and input dependent characterization of the disturbance in order to reduce the conservatism, compare~\cite{Robust_TAC_19}.

\subsection{Continuous-time dynamics} 
\label{sec:app_cont}
In the following, we summarize the continuous-time analog of the reference generic offline computations in Section~\ref{sec:loc_stab}. 
The nonlinear continuous-time dynamics are given by 
\begin{align*}
\dfrac{d}{dt}[x]=\dot{x}=f(x,u)
\end{align*}
and $f$ is assumed to be twice continuously differentiable. 
The following condition characterizes the admissible reference trajectories as the continuous-time analog of Assumption~\ref{ass:ref}. 
\begin{assumption}
\label{ass:ref_cont}
The reference signal $r:\mathbb{R}\rightarrow\mathbb{R}^{n+m}$ is continuously differentiable and satisfies
\begin{align*}
r(t)\in&\mathcal{Z}_r\subseteq\text{int}(\mathcal{Z}),\\
\dot{r}(t)\in&\mathcal{R}(r(t))=\{(\dot{x}_r,\dot{u}_r)|~\dot{x}_r=f(x_r,u_r),~\|\dot{u}_r\|_{\infty}\leq u_{\max}\},
\end{align*}
for all $t\geq 0$ with some constant $u_{\max}$.  
\end{assumption}
\begin{remark}
This assumption can be generalized to consider non-differentiable reference signal $r$ ($\dot{u}_r$ unbounded). 
In this case, the terminal cost $P_f$ should be parameterized with parameters $\theta_i$ independent of $u_r$, i.e., $P_f(x_r)$. 
\end{remark}

The following assumption characterizes the terminal ingredients, as a continuous-time analog of Assumption~\ref{ass:term}.  
\begin{assumption}
\label{ass:term_cont} 
There exist matrices $K_f(r)\in\mathbb{R}^{m\times n}$, $P_f(r)\in\mathbb{R}^{n\times n}$ with $c_l I_n\leq P_f(r)\leq c_u I_n$, $P_f$ continuously differentiable, a terminal set $\mathcal{X}_f(r)=\{x\in\mathbb{R}^n|~V_f(x,r)\leq \alpha\}$ with the terminal cost $V_f(x,r)=\|x-x_r\|_{P_f(r)}^2$, such that the following properties hold for any $r\in\mathcal{Z}_r$, any $x\in\mathcal{X}_f(r)$ and any $\dot{r}\in\mathcal{R}(r)$
\begin{align}
\label{eq:term_dec_cont}
\dfrac{d}{dt}[V_f(x,r)]\leq& -\ell(x,k_f(x,r),r),\\
\label{eq:term_con_cont}
(x,k_f(x,r))\in&\mathcal{Z},
\end{align}
with positive constants $c_l,~c_u,~\alpha$ and
\begin{align*}
\dot{x}=&f(x,k_f(x,r)), \quad k_f(x,r)=u_r+K_f(r)\cdot (x-x_r),\\
&\dfrac{d}{dt}V_f(x,r)\\
=&2(x-x_r)^\top P_f(r)(\dot{x}-\dot{x}_r)+ \|x-x_r\|_{\frac{d}{dt}{P}_f(r)}^2. 
\end{align*}
\end{assumption}
The following Lemma provides sufficient conditions for Assumption~\ref{ass:term_cont} to be satisfied based on the linearization, as a continuous-time version of Lemma~\ref{lemma:lpv}. 
\begin{lemma}
\label{lemma:lpv_cont}
Assume that there exist matrices $K_f(r)\in\mathbb{R}^{m\times n}$ continuous in $r$ and a positive definite matrix $P_f(r)\in\mathbb{R}^{n\times n}$  continuously differentiable with respect to $r$, such that for any  $r\in\mathcal{Z}_r$, $\dot{r}\in\mathcal{R}(r)$, the following matrix inequality is satisfied
\begin{align}
\label{eq:lpv_cont}
&(A(r)+B(r)K_f(r))^\top P_f(r)+ P_f(r)(A(r)+B(r)K_f(r))\nonumber\\
&+\sum_{j=1}^{n+m}\dfrac{\partial P_f}{\partial r_j}\dot{r}_j +(Q+\epsilon I_n+K_f(r)^\top R K_f(r))\leq 0
\end{align}
 with some positive constant $\epsilon$. 
Then there exists a sufficiently small constant $\alpha$, such that $P_f,~K_f$ satisfy Assumption~\ref{ass:term_cont}.  
\end{lemma}
\begin{proof}
Denote $\Delta x=x-x_r$, $\Delta u=K_f(r)\Delta x$. 
Using a first order Taylor approximation at $r=(x_r,u_r)$, we get
\begin{align*}
f(x,u)={f(x_r,u_r)}+A(r)\Delta x+B(r)\Delta u+\Phi_r(\Delta x),
\end{align*}
with the remainder term $\Phi_r$. 
The terminal cost satisfies
\begin{align*}
&\dfrac{d}{dt} V_f(x,r)=2(x-x_r)^\top P_f(r)(\dot{x}-\dot{x}_r)\\
&+(x-x_r)^\top \left[\sum_{j=1}^{n+m} \dfrac{\partial P_f}{\partial r_j} \dot{r_j}\right](x-x_r)\\
\stackrel{\eqref{eq:lpv_cont}}{\leq} &-\ell(x,k_f(x,r))-\epsilon\|\Delta x\|^2+2(x-x_r)^\top P_f(r) \Phi_r(\Delta x).
\end{align*}
For $\alpha$ sufficiently small, this implies~\eqref{eq:term_dec_cont} (due to the arbitrarily small local Lipschitz bound on the higher order terms $\Phi_r$). 
Constraint satisfaction~\eqref{eq:term_con_cont} is guaranteed analogous to Lemma~\ref{lemma:lpv}. 
\end{proof}
The following Lemma provides corresponding LMI conditions, similar to Lemma~\ref{lemma:LMI}. 
\begin{lemma}
\label{lemma:LMI_cont}
Suppose that there exists a matrix $Y(r)$  continuous in $r$  and $X(r)$ continuously differentiable with respect to $r$, that satisfy the constraints in~\eqref{eq:LMI_cont} for all $r\in\mathcal{Z}_r,~\dot{r}\in\mathcal{R}(r)$. 
Then  $P_f=X^{-1}$, $K_f=YP_f$ satisfy~\eqref{eq:lpv_cont}.
\end{lemma}
\begin{proof}
Multiplying~\eqref{eq:lpv_cont} from left and right with $X(r)$ yields
\begin{align*}
&(A(r)X(r)+B(r)Y(r))^\top +(A(r)X(r)+B(r)Y(r)) \\
&+X(r)\dfrac{d}{dt}[X^{-1}(r)]X(r)+X(r)(Q+\epsilon I_n)X(r)\\
&+Y(r)^\top R Y(r)\leq 0.
\end{align*}
Note that the chain rule applied to the inverse of $X$ yields 
\begin{align*}
X(r)\dfrac{d}{dt}\left[X^{-1}(r)\right]X(r)=-\dfrac{d}{dt}[X(r)]=-\sum_{i=1}^p X_i\dfrac{\partial \theta_i}{\partial r} \dot{r} .
\end{align*}
Applying the Schur complement results in~\eqref{eq:LMI_cont}. 
\end{proof}
If a gridding approach is considered to compute the terminal ingredients, one needs to grid $r\in\mathcal{Z}_r$ and consider the $2 ^m$ vertices of $\dot{u}_r$  (since~\eqref{eq:LMI_cont} is affine in $\dot{u}_r$).

For the convex approach, polytopic sets need to be constructed such that 
$(\theta,\dot{\theta})\in\Theta\times\Omega=\overline{\Theta},~\forall r\in\mathcal{Z}_r,~\dot{r}\in\mathcal{R}(r)$.
The following proposition provides the corresponding LMI conditions based on the vertices of $\overline{\Theta}$, similar to Proposition~\ref{prop:LMI_lpv}. 
\begin{proposition}
\label{prop:LMI_lpv_cont}
Suppose that there exists matrices $X_i,~Y_i,~\Lambda_i,~X_{\min}$ that satisfy the constraints in~\eqref{eq:LMI_LPV_cont}. 
Then the following matrices satisfy~\eqref{eq:lpv_cont}: 
\begin{align*}
P_f(r)=&X^{-1}(r),\quad K_f(r)=Y(r)P_f(r).
\end{align*} 
\end{proposition}
\begin{proof}
The proof is analogous to Proposition~\ref{prop:LMI_lpv}, based on multi-convexity and Lemma~\ref{lemma:LMI_cont}.
\end{proof}

\begin{remark}
\label{rk:sample_hold}
The continuous-time formulation is suitable if the online MPC optimization considers continuous-time input signals, instead of piece-wise constant inputs (as is common in many numerical implementations).  
Nevertheless, if the sampling time $h$ is sufficiently small, the continuous-time terminal cost  (scaled by $1/h$) might satisfy the discrete-time conditions (Ass.~\ref{ass:term}) with a piece-wise constant input. 
This can be favorable since the corresponding offline optimization problem~\eqref{eq:LMI_cont} or \eqref{eq:LMI_LPV_cont} is often easier formulated and faster solved, especially if a non-trivial discretization is considered.  
Algorithm~\ref{alg:offline_alpha} can be used to ensure the validity of the computed terminal ingredients with the zero-order hold input (instead of the continuous-time feedback). 
\end{remark}

\subsection{Output tracking stage cost}
\label{sec:app_output}
In the following, we discuss how the derivation in Section~\ref{sec:loc_stab} can be extended to deal with an output tracking stage cost. 
As an alternative to~\eqref{eq:stage}, consider the following output reference tracking stage cost
\begin{align}
\label{eq:stage_output}
\ell(x,u,r)=\|h(x,u)-h(x_r,u_r)\|_{S(r)}^2,
\end{align}
with a nonlinear twice continuously differentiable output function $h:\mathcal{Z}\rightarrow \mathbb{R}^p$ and a positive definite weighting matrix $S(r)$, which assumed to be continuous in $r$. 
Such a stage cost can be used for output regulation, output trajectory tracking, output path following or manifold stabilization, compare~\cite{faulwasser2012optimization}. 
We denote the Jacobian of the output $h$ around an arbitrary point $r\in\mathcal{Z}_r$ by
\begin{align}
\label{eq:C_r}
C(r)=\left.\left[\dfrac{\partial h}{\partial x}\right]\right|_{(x,u)=r},\quad D(r)=\left.\left[\dfrac{\partial h}{\partial u}\right]\right|_{(x,u)=r}. 
\end{align}
The following lemma establishes sufficient conditions for Assumption~\ref{ass:term} with the stage cost~\eqref{eq:stage_output} based on the linearization, similar to Lemma~\ref{lemma:lpv}. 
\begin{lemma}
\label{lemma:lpv_output}
Suppose that $f,~h$ are twice continuously differentiable.
Assume that there exists a matrix $K_f(r)\in\mathbb{R}^{m\times n}$ and a positive definite matrix $P_f(r)\in\mathbb{R}^{n\times n}$ continuous in $r$, such that for any  $r\in\mathcal{Z}_r$, $r^+\in\mathcal{R}(r)$, the following matrix inequality is satisfied
\begin{align}
\label{eq:lpv_output}
&(A(r)+B(r)K_f(r))^\top P_f(r^+)(A(r)+B(r)K_f(r))- P_f(r)\\
\leq&-(C(r)+D(r)K_f(r))^\top S(r) (C(r)+D(r)K_f(r))-\tilde{\epsilon} I_n\nonumber
\end{align}
 with some positive constant $\tilde{\epsilon}$. 
Then there exists a sufficiently small constant $\alpha$, such that $P_f,~K_f$ satisfy Assumption~\ref{ass:term}.  
\end{lemma}
\begin{proof}
A first order Taylor approximation at $r=(x_r,u_r)$ yields
\begin{align*}
&h(x,k_f(x,r))-{h(x_r,u_r)}\\
=&(C(r)+D(r)K_f(r))\Delta x+\tilde{\Phi}_{r}(\Delta x),
\end{align*}
with the remainder term $\tilde{\Phi}_{r}$ and $\Delta x=x-x_r$.  
The stage cost satisfies
\begin{align}
\label{eq:output_1}
&\ell(x,k_f(x,r),r) 
\\ 
\geq &\|(C(r)+D(r)K_f(r))\Delta x\|_{S(r)}^2+\|\tilde{\Phi}_{r}(\Delta x)\|_{S(r)}^2\nonumber\\
&- 2\|\tilde{\Phi}_{r}(\Delta x)\|_{S(r)} \|(C(r)+D(r)K_f(r))\Delta x \|_{S(r)}.\nonumber
\end{align}
Given continuity and compactness, there exists a constant 
\begin{align}
\label{eq:output_2}
c_y=\max_{r\in\mathcal{Z}_r}\|(C(r)+D(r)K_f(r))\|_{S(r)}.
\end{align}
For a sufficiently small $\alpha$, the remainder term $\tilde{\Phi}_{r}$  satisfies the following (local) Lipschitz bound
\begin{align}
\label{eq:output_3}
\|\tilde{\Phi}_{r}(\Delta x)\|_{S(r)}/\|\Delta x\|=:\tilde{L}_{r,x}\leq \tilde{L}^*:=c_y-\sqrt{c_y^2-\tilde{\epsilon}/2},
\end{align}
for all $x\in\mathcal{X}_f(r)$ and all $r\in\mathcal{Z}_r$. 
This implies
\begin{align*}
&\ell(x,k_f(x,r))\\
\stackrel{\eqref{eq:output_1},\eqref{eq:output_3}}{\geq}& \|(C(r)+D(r)K_f(r))\Delta x\|_{S(r)}^2+\tilde{L}_{r,x}^2\|\Delta x\|^2\\
&-2\tilde{L}_{r,x}\|\Delta x\|^2\|(Cr)+D(r)K_f(r)\|_{S(r)}\\
\stackrel{\eqref{eq:output_2}}{\geq}& \|(C(r)+D(r)K_f(r))\Delta x\|_{S(r)}^2+\tilde{L}_{r,x}(\tilde{L}_{r,x}-2c_y)\|\Delta x\|^2\\
\stackrel{\eqref{eq:output_3}}{\geq}& \|(C(r)+D(r)K_f(r))\Delta x\|_{S(r)}^2+\tilde{L}^*(\tilde{L}^*-2c_y)\|\Delta x\|^2\\
\stackrel{\eqref{eq:output_3}}=& \|(C(r)+D(r)K_f(r))\Delta x\|_{S(r)}^2-\tilde{\epsilon}/2\|\Delta x\|^2.
\end{align*}
The second to last step follows by using the fact that the function $L(L-2c_y)$ attains it minimum for $L\in[0,L^*]$ at $L=L^*$. 
Combining the derived bound on $\ell(x,k_f(x,r))$ with~\eqref{eq:lpv_output} ensures that the terminal cost $V_f$ satisfies inequality~\eqref{eq:lpv_1} in Lemma~\ref{lemma:lpv} with the modified stage cost and with ${\epsilon}=\tilde{\epsilon}/2$. 
The remainder of the proof is analogous to Lemma~\ref{lemma:lpv}. 
\end{proof}
\begin{remark}
\label{rk:output_1}
For the linear output $h(x,u)=[Q^{1/2}x;R^{1/2}u]\in\mathbb{R}^{n+m}$ and $S=I$ we recover the conditions in Lemma~\ref{lemma:lpv} with the stage cost~\eqref{eq:stage}.  
\end{remark}
\begin{remark}
\label{rk:output_2}
Depending on the output $h$ and the reference $r$, there may exist multiple solutions that achieve exact output tracking. 
Thus, we can in general not expect asymptotic/exponential stability of the reference $r$, but instead stability of a corresponding set or manifold, compare~\cite{faulwasser2012optimization}. 
Under suitable (incremental) detectability conditions on the output $h$, we can recover stability of the specific reference trajectory $r$. 
\end{remark}
Based on these conditions, Lemma~\ref{lemma:LMI_output} provides LMI conditions to compute $P_f,~K_f$, similar to Lemma~\ref{lemma:LMI}. 
Furthermore, if the parameters $\theta_i$ are chosen, such that
\begin{align*}
S^{-1}(r)=S_0+\sum_{i=1}^p \theta_i(r) S_i,
\end{align*}
then Proposition~\ref{prop:LMI_lpv_output} yields LMI conditions based on the vertices of $\overline{\Theta}$, similar to Proposition~\ref{prop:LMI_lpv}.

\begin{table*}
\small{
\begin{subequations}
\label{eq:LMI_cont}
\begin{align}
\min_{X(r),Y(r),X_{\min}}& - \log\det X_{\min}\\
\text{s.t. }&\begin{pmatrix}
A(r)X(r)+B(r)Y(r) +(A(r)X(r)+B(r)Y(r))^\top -\dfrac{d}{dt}[X(r)]&((Q+\epsilon)^{1/2}X(r))^\top&(R^{1/2}Y(r))^\top\\
*&-I&0\\
*&0&-I\\
\end{pmatrix}\leq 0,\\
&X_{\min}\leq X(r),\\
&\forall r\in\mathcal{Z}_r,~\dot{r}\in\mathcal{R}(r).  
\end{align}
\end{subequations}
}
\label{tab:long-eq}
\end{table*} 
\begin{table*}
\small{
\begin{subequations}
\label{eq:LMI_LPV_cont}
\begin{align}
\min_{X_i,Y_i,\Lambda_i,X_{\min}}& - \log\det X_{\min}\\
\text{s.t. }&\begin{pmatrix}
A(\theta)X(\theta)+B(\theta)Y(\theta) +(A(\theta)X(\theta)+B(\theta)Y(\theta))^\top -X(\dot{\theta})+X_0&((Q+\epsilon)^{1/2}X(\theta))^\top&(R^{1/2}Y(\theta))^\top\\
*&-I&0\\
*&0&-I\\
\end{pmatrix}\nonumber\\
\leq& -  \begin{pmatrix}\sum_{i=1}^p\theta_i^2\Lambda_i&0\\0&0\end{pmatrix},\\
&X_{\min}\leq X(\theta),\quad \forall (\theta,\dot{\theta})\in\text{Vert}(\overline{\Theta}),  \\
\label{eq:LMI_LPV2_cont}
&\Lambda_i+
(A_iX_i+B_iY_i)
+(A_iX_i+B_iY_i)^\top
\geq 0,\quad \Lambda_i\geq 0,\quad  i=1,\dots,p.
\end{align}
\end{subequations}
}
\end{table*}

\begin{table*}
\begin{lemma}
\label{lemma:LMI_output}
Suppose that there exists matrices $X(r)$,~$Y(r)$ continuous in $r$, that satisfy the following constraints
\small{
\begin{subequations}
\label{eq:LMI_output}
\begin{align}
\min_{X(r),Y(r),X_{\min}}& - \log  \det X_{\min}\\
\text{s.t. }&\begin{pmatrix}
X(r)&(A(r)X(r)+B(r)Y(r))^\top&(C(r)X(r)+D(r)Y(r))^\top&\sqrt{\tilde{\epsilon}} X(r)\\
*&X(r^+)&0&0\\
*&*&S^{-1}(r)&0\\
*&*&*&I
\end{pmatrix}\geq 0,\\
&X_{\min}\leq X(r),\quad 
\forall r\in\mathcal{Z}_r,~r^+\in\mathcal{R}(r).  
\end{align}
\end{subequations}
}
\label{tab:long-eq}
\normalsize
Then  $P_f=X^{-1}$, $K_f=YP_f$ satisfy~\eqref{eq:lpv_output}.
\end{lemma}
\begin{proof}
The proof is similar to Lemma~\ref{lemma:lpv}, compare also~\cite{boyd1994linear}. 
Define $X(r)=P_f(r)^{-1}$ and $Y(r)=K_f(r)X(r)$. 
Multiplying~\eqref{eq:lpv_output} from left and right with $X(r)$ yields
\small{
\begin{align*}
&(A(r)X(r)+B(r)Y(r))^\top X(r^+)^{-1} (A(r)X(r)+B(r)Y(r))-X(r)
+\tilde{\epsilon} X(r) IX(r)\\
&+(C(r)X(r)+D(r)Y(r))^\top S(r)(C(r)X(r)+D(r)Y(r))\leq 0.
\end{align*}
}
\normalsize
This can be equivalently written as
\small{
\begin{align*}
X(r)-
\begin{pmatrix}
A(r)X(r)+B(r)Y(r)\\C(r)X(r)+D(r)Y(r)\\\sqrt{\tilde{\epsilon}}X(r)
\end{pmatrix}^{\top}
\begin{pmatrix}
X(r^+)^{-1}&0&0\\
0&S(r)&0\\
0&0&I
\end{pmatrix}
\begin{pmatrix}
A(r)X(r)+B(r)Y(r)\\C(r)X(r)+D(r)Y(r)\\\sqrt{\tilde{\epsilon}}X(r)
\end{pmatrix}
\geq 0
\end{align*}}
\normalsize
Using the Schur complement this reduces to~\eqref{eq:LMI_output}, which is linear in $X,~Y$. 
\end{proof}
\end{table*}

\begin{table*}
\begin{proposition}
\label{prop:LMI_lpv_output}
Suppose that there exists matrices $X_i,~Y_i,~\Lambda_i,~X_{\min}$ that satisfy the following constraints
\small{
\begin{subequations}
\label{eq:LMI_LPV_output}
\begin{align}
\min_{X_i,Y_i,\Lambda_i,X_{\min}}& - \log\det X_{\min}\\
\text{s.t. }&\begin{pmatrix}
X(\theta)&(A(\theta)X(\theta)+B(\theta)Y(\theta))^\top&(C(\theta)X(\theta)+D(\theta)Y(\theta))^\top&\sqrt{\tilde{\epsilon}} X(\theta)\\
*&X(\theta^+)&0&0\\
*&*&S^{-1}(\theta)&0\\
*&*&*&I
\end{pmatrix}\geq   \begin{pmatrix}\sum_{i=1}^p\theta_i^2\Lambda_i&0\\0&0\end{pmatrix},\\
&X_{\min}\leq X(\theta),\quad \forall (\theta,\theta^+)\in\text{Vert}(\overline{\Theta}),  \\
\label{eq:LMI_LPV_output2}
&\begin{pmatrix}
0&(A_iX_i+B_iY_i)^\top&(C_iX_i+D_iY_i)^\top\\
(A_iX_i+B_iY_i)&0&0\\
(C_iX_i+D_iY_i)&0&0
\end{pmatrix}\leq \Lambda_i,\quad \Lambda_i\geq 0,\quad  i=1,\dots,p.
\end{align}
\end{subequations}
}
\normalsize
Then $P_f=X^{-1}$ and $K_f=Y P_f$ satisfy~\eqref{eq:lpv_output}.
\end{proposition}
\begin{proof}
The proof is analogous to Proposition~\ref{prop:LMI_lpv} based on Lemma~\ref{lemma:LMI_output}. 
The constraint~\eqref{eq:LMI_LPV_output2} ensures multi-convexity. 
\end{proof}
\end{table*}

\end{document}